\documentclass[11pt, leqno]{amsart}

\usepackage{amsxtra}
\usepackage{graphicx}
\usepackage{newlfont}
\usepackage{amscd}
\usepackage{hyperref}
\usepackage[german,english]{babel}
\usepackage{cancel}
\usepackage{amsmath,amsthm}
\usepackage{amssymb}
\usepackage{enumerate}
\usepackage{color}
\usepackage[all,cmtip]{xy}

 \newtheorem{thm}{Theorem}[section]

 \newtheorem{cor}[thm]{Corollary}
 \newtheorem{lem}[thm]{Lemma}
 \newtheorem{prop}[thm]{Proposition}
 
  \theoremstyle{definition}
 \newtheorem{defn}[thm]{Definition}
  \newtheorem{note}[thm]{Note}
  \newtheorem{defn-thm}[thm]{Definition-Theorem}
 \theoremstyle{remark}

 \newtheorem{rem}[thm]{Remark}
 \newtheorem{ex}[thm]{Example}

\numberwithin{equation}{section}

\numberwithin{thm}{section}

\numberwithin{table}{section}

\numberwithin{figure}{section}

\ifx\pdfoutput\undefined
  \DeclareGraphicsExtensions{.pstex, .eps}
\else
  \ifx\pdfoutput\relax
    \DeclareGraphicsExtensions{.pstex, .eps}
  \else
    \ifnum\pdfoutput>0
      \DeclareGraphicsExtensions{.pdf}
    \else
      \DeclareGraphicsExtensions{.pstex, .eps}
    \fi
  \fi
\fi

\newcommand{\ZZ}{\mathbb{Z}}

\newcommand{\RR}{\mathbb{R}}

\newcommand{\CC}{\mathbb{C}}

\newcommand{\g}{\mathfrak{g}}

\newcommand{\m}{\mathfrak{m}}
\newcommand{\n}{\mathfrak{n}}

\newcommand{\vac}{\left| 0 \right>}

\newcommand{\ad}{\text{ad}}
\newcommand{\gr}{\text{gr}}
\newcommand{\sll}{\text{sl}}
\newcommand{\WW}{\mathcal{W}}

\begin{document}

\title{ Classcal affine W-algebras associated to Lie superalgebras }
\author{ Uhi Rinn Suh$^{1}$}
\address{Department of Mathematical Sciences, Seoul National University, GwanAkRo 1, Gwanak-Gu, Seoul 151-747, Korea}
\email{uhrisu1@math.snu.ac.kr}
\thanks{$^{1}$This work was supported by BK21 PLUS SNU Mathematical Science Division.}
\maketitle

\begin{abstract}
In this paper, we prove classical affine W-algebras associated to Lie superalgebras (W-superalgebras) can be constructed in two different ways: via affine classical Hamiltonian reductions and via taking quasi-classical limits of quantum affine W-superalgebras. Also, we show that a classical finite W-superalgebra can be obtained by a Zhu algebra of a classical affine W-superalgebra. Using the definition by Hamiltonian reductions, we find free generators of a classical W-superalgebra associated to a minimal nilpotent. Moreover, we compute generators of the classical W-algebra associated to $spo(2|3)$ and its principal nilpotent. In the last part of this paper, we introduce a generalization of classical affine W-superalgebras called classical affine fractional W-superalgebras. We show these have Poisson vertex algebra structures and find generators of a fractional W-superalgebra associated to a minimal nilpotent. 
\end{abstract}

\setcounter{tocdepth}{-1}

\pagestyle{plain}

\section{Introduction}\label{Sec:Introduction}

This paper is a generalization of \cite{S, S2}, which showed equivalences of various definitions of classical affine W-algebras and introduced classical fractional W-algebras. 

Recall that there are four types of W-algebras: classical affine, classical finite, quantum affine and quantum finite W-algebras. These types of algebras are endowed with Poisson vertex algebras (PVAs), Poisson algebras (PAs), vertex algebras (VAs) and associative algebras (AAs) structures, respectively.  As underlying algebraic structures in mathematical physics, PVAs (resp. PAs) are quasi-classical limits of VAs (resp. AAs) and Poisson algebras (resp. AAs) are finalizations of PVAs (resp. VAs). (See \cite{DK, DKV1, K, Z}.)

The main ingredient of this paper is a classical affine W-algebra, which is endowed with PVA structures. Hence we expect classical affine W-algebras are obtained by quasi-classical limits of quantum affine W-algebras and chiralizations of classical finite W-algebras. A classical finite W-algebra $\WW^{fin}(\g,f)$ associated to a Lie (super)algebra $\g$ and its nilpotent $f$ is defined by the Hamiltonian reduction
\[ \WW^{fin}(\g,f)= (S(\g)/S(\g) I)^{\ad \n}\]
 associated to $(S(\g), S(\g)I, \n)$ for a Poisson (super)algebra ideal $S(\g)I$ and a nilpotent Lie  subalgebra $\n$ of $\g$ determined by $f$. Also, there is an equivalent  construction of $\WW^{fin}(\g,f)$ by a cohomology of Lie (super)algebra complex. 

A natural way to get a quantum finite W-algebra is by the BRST quantization of the Lie (super)algebra complex, called a finite BRST complex. As in classical finite W-algebras cases,  in \cite{DK, GG}, it is proved that the quantum finite W-algebra $W^{fin}(\g,f)$ associated to a Lie (super)algebra $\g$ and its nilpotent $f$ can be obtained by a quantum Hamiltonian reduction associated to $U(\g)$, its associative algebra ideal $U(\g)I$ and a nilpotent Lie  subalgebra $\n$ determined by $f$.

In \cite{DK, FF, KRW, KW}, the quantum affine W-algebra $W(\g,f,k)$ is introduced by BRST complex which is obtained by substituting universal enveloping algebras of Lie (super)algebras in the finite BRST complex with universal enveloping vertex algebras of Lie conformal algebras (LCAs).  

In \cite{S2}, by substituting universal enveloping vertex algebras of LCAs in the BRST complexes with symmetric algebras generated by the LCAs, we get classical affine W-algebras. For the classical affine W-algebra $\WW(\g,f,k)$ associated to a Lie algebra $\g$, there is an equivalent definition via an affine Hamiltonian reduction. Also, two W-algebras $\WW^{fin}(\g,f)$ and $\WW(\g,f,k)$ are related by a finitization map called Zhu map. (See \cite{DKV1}.)

A natural question is that if we can develop a similar theory for a classical affine W-algebra associated to a Lie superalgebra (classical affine W-superalgebra). In Section \ref{Sec:superW}, we prove that a classical affine W-superalgebra can be defined via classical BRST complex and via Hamiltonian reduction. Also, we show that the same argument works for classical and quantum finite W-superalgebras. Moreover, in Section \ref{Sec:finite}, we describe relations between affine W-superalgebras and finite W-superalgebras.

Also, structure theories of finite W-superalgebras are developed in various articles, for example \cite{Pol, PS, Zhao}. In this paper, we investigate structures of classical affine W-superalgebras. 

The simplest example of W-superalgebras can be obtained by taking a minimal nilpotent $f$ of given Lie superalgebra $\g$. In \cite{KW}, Kac and Wakimoto discovered free generators of quantum affine W-(super)algebras associated to minimal nilpotents and Premet \cite{P1} described generators of finite W-algebras associated to Lie algebras and minimal nilpotents. In \cite{S}, similar results are written for classical affine W-algebras associated to Lie algebras and their minimal nilpotents. In Section \ref{Sec:Example}, we show that generators of a classical affine W-superalgebra associated to a minimal nilpotent can also be described explicitly. In addition, we compute $\lambda$-brackets between the generators. 

It is still open what are free generators of classical affine W-superalgebras associated to non-minimal nilpotents. However, it is possible to find free generators for simple cases by computations.  In Section \ref{Sec:Example}, we find free generators of classical affine W-algebras associated to $\g=spo(2|3)$ and its principal nilpotent.   The Poisson $\lambda$-brackets between the generators are also computed directly.

The last part of this paper is about fractional W-superalgebras. In \cite{BDHM} and \cite{DHM}, they introduced fractional W-algebras as a generalization of W-algebras in \cite{DS}. Note that W-algebras in \cite{DS}  appear as underlying algebraic structures of integrable systems  and they are isomorphic to classical affine W-algebras associated to Lie algebras in our context. Similarly, fractional W-algebras are also related to integrable systems. In \cite{S2}, PVAs called classical affine fractional W-algebras associated to Lie algebras are introduced which are isomorphic to fractional W-algebras in \cite{BDHM}.

In Section \ref{Sec:frac}, we define classical affine fractional W-algebras associated to Lie superalgebras (fractional W-superalgebras). In \cite{S2}, the well-definedness of a fractional W-algebra as a PVA is proved by the fact that it is isomorphic to a fractional W-algebra in \cite{BDHM}. However, since it is not clear how to construct fractional W-superalgebras in the context of \cite{BDHM}, we cannot use the same argument as in \cite{S2}. In this paper, we prove that classical affine fractional W-algebras are well-defined PVAs with a simpler method and show that the proof works for fractional W-superalgebras.
Moreover, we find free generators of a classical affine fractional W-superalgebra associated a minimal nilpotent and also Poisson $\lambda$-brackets between them. 

\section*{Acknowledgement}

The author would like to thank her thesis advisor Victor Kac and Tomoyuki Arakawa for valuable discussions during the INdAM Intensive Period: Perspectives in Lie Theory held in Pisa, Italy.

\vskip 3mm

\section{Backgrounds} \label{Sec:background}
In Section \ref{Sec:background}, we recall some basic notions that we need to investigate W-algebras. 
\subsection{Poisson vertex algebras: relations with vertex algebras and Poisson algebras}\ \\
A vector space $V$ over $\CC$ with a decomposition $V=V_{\bar{0}}\oplus V_{\bar{1}}$ is called a {\it vector superspace}. A vector in $V_{\bar{0}}$ (resp. $V_{\bar{1}}$) is said to be  {\it even} (resp. {\it odd}). An element in $V_{\bar{0}}$ or $V_{\bar{1}}$ is called a {\it homogeneous} element. The {\it parity} $p(v)$ of a homogeneous element $v\in V_{\bar{0}}$ (resp.  $v\in V_{\bar{1}}$) is $0$ (resp. $1$). 


An  algebra $A$ over $\CC$ is called a {\it superalgebra} if it is a $\ZZ/2\ZZ$-graded algebra. In other words, $A=A_{\bar{0}}\oplus A_{\bar{1}}$ as a vector space and $A_{\bar{i}}A_{\bar{j}} \subset A_{\bar{i}+\bar{j}}$.

A vector superspace $V$ endowed with a bracket $[ \, , \, ]:V\times V\to V$ is called a {\it Lie superalgebra} if it satisfies \\
\begin{equation*}
\begin{aligned}
& \text{(bilinearity) $[ka+b,c]=k[a,c]+[b,c], \quad a,b,c\in V, \quad k\in \CC;$ }\\
& \text{(skewsymmetry) $ [a, b]=-(-1)^{p(a)p(b)}[b,a], \quad a,b\in V;$ }\\
& \text{(Jacobi identity) $[a,[b,c]]=[[a,b],c]]+(-1)^{p(a)p(b)}[b,[a,c]], \quad a,b,c\in V.$}
\end{aligned}
\end{equation*}

A {\it Poisson superalgebra } $P$ endowed with the bracket $\{\, , \, \}$ satisfies the following properties:
\begin{enumerate}
\item $(P, \{\, , \, \})$ is a Lie superalgebra.
\item $P$ is a supersymmetric algebra, that is $ab=(-1)^{p(a)p(b)}ba$ for $a,b\in P$
\item $\{a,bc\}= (-1)^{p(a)p(b)}b\{a,c\}+\{a,b\}c$ for $a,b,c\in P.$ 
\end{enumerate}

Now we recall differential algebras and Lie conformal algebras, which are needed to introduce vertex algebras and  Poisson vertex algebras. The definitions can be found in \cite{BK, K}.

\begin{defn}
\begin{enumerate}
\item A superalgebra $A$ is called a {\it differential algebra} if it is endowed with a parity preserving map $\partial : A \to A$ satisfying $\partial(ab)=(\partial a)b+a(\partial b).$
\item
Let $R$ be a vector superspace over $\CC$ with a $\CC[\partial]$-module structure and let $\partial$ be a parity preserving map on $R$. 
The $\CC[\partial]$-module $R$ endowed with a linear $\lambda$-bracket $[\, _\lambda \, ]: R\otimes_\CC R \to R[\lambda]$ is called a {\it Lie conformal algebra (LCA)}  if it satisfies 
\begin{equation*}
\begin{aligned}
& \text{(sesquilinearity) $[\partial a_\lambda b]=-\lambda [a_\lambda b]$, $[a_\lambda \partial b]=(\partial+\lambda)[a_\lambda b]$,  $ \ \ a,b\in R;$ }\\
& \text{(skewsymmetry) $ [a _\lambda b]=-(-1)^{p(a)p(b)}[b_{-\partial-\lambda}a], \quad a,b\in R;$ }\\
& \text{(Jacobi identity) $[a_\lambda [b_\mu c]]=[[a_\lambda b]_{\lambda+\mu} c]+(-1)^{p(a)p(b)}[b_\mu[a_\lambda c]], \qquad a,b,c\in R.$}
\end{aligned}
\end{equation*}
\end{enumerate}
\end{defn}

Let $R$ be a LCA and let $a,b$ be elements in $R$. There are $c_n\in R$ for $n\in \ZZ_{\geq 0}$ such that $[a_\lambda b]=\sum_{n\in \ZZ_{\geq 0}} \frac{c_n}{n!} \lambda^n$. We denote $c_n$ by $a_{(n)}b.$

\begin{defn} \cite{BK} \label{Def:VA}
A quintuple $(V, \partial, \vac, : \, \, :, [\, _\lambda\, ])$ is called a {\it vertex algebra} if it satisfies
\begin{enumerate}
\item $(V, \partial, [\, _\lambda \, ])$ is a Lie conformal algebra, 
\item $(V, \partial, \vac, :\, \, :)$ is a differential algebra with the strong quasicommutativity,
\item the $\lambda$-bracket $[ \,  _\lambda \, ]$ and the normally ordered product $:\, \, :$ are related by the noncommutative Wick formula,
\end{enumerate}
where the strong quasicommutativity is 
\[ :a:bc:-(-1)^{p(a)p(b)}:b:ac:: \, = \, : \int_{-\partial}^0[a_\lambda b] d\lambda \, c: \, = \, -\sum_{n\in \ZZ_+} : \, \frac{(-1)^{n+1}}{(n+1)!}(\partial^{n+1} a_{(n)}b) \ c : \]
and the noncommutative Wick formula is 
\[ [a_\lambda :bc:]=:[a_\lambda b] c:+ (-1)^{p(a)p(b)}:[b_\lambda a] c: + \int_0^\lambda[[a_\lambda b]_\mu c] d\mu, \]
for $a,b,c\in R.$
Here $\int_0^\lambda[[a_\lambda b]_\mu c] d\mu = \sum_{n,m\in \ZZ_+} \frac{((a_{(n)}b)_{(m)} c)}{n!(m+1)!}\lambda^{n+m+1}.$
\end{defn}

As in the Lie superalgebra theory, there is a unique universal enveloping vertex algebra $(V(R), i: R \to V(R))$ of a Lie conformal algebra $R$. The universality is that  if there is a vertex algebra $V$ endowed with a LCA homomorphism $r:R\to V$ then there is a unique vertex algebra homomorphism $q: V(R) \to V$ such that $q\circ i=r.$

There is the PBW theorem of universal enveloping vertex algebras, in the sense of Theorem \ref{PBW:VA}.

\begin{thm} \cite{K} \label{PBW:VA}
Let $(V(R), i: R \to V(R))$ be the universal enveloping vertex algebra of a Lie conformal algebra $R$.
\begin{enumerate}
\item The map $i:R\to V(R)$ is injective.
\item If $\{u_1, \cdots, u_k\}$ is a $\CC$-basis of $R$ then $\{\, :u_{i_1} u_{i_2}\cdots u_{i_l}:\, | \, 1\leq i_1\leq i_2\leq \cdots\leq i_l \leq k\, \}$ is a $\CC$-basis of $V(R),$ where $:u_{i_1} u_{i_2}\cdots u_{i_l}: $ denotes normally ordered products from right to left, that is  \[u_{i_1} u_{i_2}\cdots u_{i_l}: =( :u_{i_1}(:u_{i_2}(:u_{i_3}( \cdots (:u_{i_{l-1}}u_{i_l}:)\cdots):):):).\] 
 Equivalently, if  $\mathcal{B}$ is a $\CC$-basis of $R$ then $V(R)$ is {\it freely generated} by $\mathcal{B}.$
\end{enumerate}
\end{thm}

\begin{ex} \label{Ex:Current}
The current Lie conformal algebra associated to the finite dimensional Lie superalgebra $\g$ with an invariant supersymmetric bilinear form $(\, |\, ):\g\times \g \to \CC$ is 
\[ Cur(\g)=\CC[\partial]\otimes \g \oplus\CC K \text{ where } p(K)=0,\]
endowed with the $\lambda$-bracket defined by 
\[ [a_\lambda b]=[a,b]+\lambda(a|b) K  \text{ for } a,b\in \g,  \quad [K_\lambda Cur(\g)]=0.\]
The universal enveloping vertex algebra $V(Cur(\g))$ of $Cur(\g)$ is freely generated by a basis $\mathcal{B}$ of $Cur(\g)$ over $\CC$ and $V^k(Cur(\g))=V(Cur(\g))/(K-k)V(Cur(\g))$ is called the universal enveloping affine vertex algebra of level $k.$
\end{ex}

\begin{defn} \cite{BK}
A quintuple $(\mathcal{V}, \partial, 1, \cdot, \{\, _\lambda \, \})$ is called a {\it Poisson vertex algebra (PVA)} if it satisfies
\begin{enumerate}
\item $(\mathcal{V}, \partial, \{ \, _\lambda\, \})$ is a Lie conformal algebra
\item $(\mathcal{V}, \partial, \cdot, 1)$ is a unital supersymmetric differential algebra.
\item the $\lambda$-bracket $\{\, _\lambda \, \}$ and the supersymmetric product are related by the Leibniz rule 
\[ \{a_\lambda bc\} = (-1)^{p(a)p(b)} b\{a_{\lambda} c\} + \{a_{\lambda}b\} c.\]
\end{enumerate}
\end{defn}

\begin{ex} \label{Ex:CurrPVA}
Let $R=Cur(\g)$ be in Example \ref{Ex:Current} and let $S(R)$ be the supersymmetric algebra generated by $R.$ We define the  $\lambda$-bracket on $S(R)$ by that on $R$ and Leibniz rules. Then $S(R)$ is a Poisson vertex algebra. Also, $S^k(R)=S(R)/(K-k)S(R)$ is a Poisson vertex algebra, for any $k\in \CC.$
\end{ex}

\begin{defn} \cite{DK}
\begin{enumerate}
\item
Consider a family of vertex algebras $V_\epsilon$ which is a vertex algebra over $\CC[\epsilon]$ with a $\lambda$-bracket such that $[V_{\epsilon \, \lambda} V_\epsilon] \subset \CC[\partial]\otimes \epsilon V_\epsilon.$ The vertex algebra $V_\epsilon$ is called {\it regular} if the multiplication by $\epsilon$ is an injective map. 
\item
Let $(V_\epsilon, \vac_\epsilon, \partial, [\, _\lambda\, ]_\epsilon, : \, \, :_\epsilon)$ be a regular family of  vertex algebras over $\CC[\epsilon].$ Let $\mathcal{V}:=V_\epsilon/\epsilon V_\epsilon$ be endowed with the product induced by the normally ordered product $: \, \, :_\epsilon$ of $V_\epsilon$ and the $\lambda$-bracket $\{ \, _\lambda\, \}$ defined by \[ \{\bar{a}_\lambda \bar{b}\}= [a_\lambda b]_\epsilon/\epsilon  \]
where $a, b \in V_\epsilon$ are preimages of $\bar{a}, \bar{b} \in \mathcal{V}$. Then $\vac_\epsilon \in V_\epsilon $ induces the unital $1\in \mathcal{V}$ and $\partial$ on $V_\epsilon$ induces a differential $\partial$ on $\mathcal{V}$. The quintuple $(\mathcal{V}, 1, \partial, \{\, _\lambda \, \}, \cdot)$ is called the {\it quasi-classical limit} of $V_\epsilon.$
\end{enumerate}
\end{defn}

It is easy to see the following remark.

\begin{rem}
The quasi-classical limit $\mathcal{V}$ of the regular family of vertex algebras $V_\epsilon$ over $\CC[\epsilon]$ is a Poisson vertex algebra. 
\end{rem}

\begin{ex} \label{Ex:Current_ClassicalLimit}
As in Example \ref{Ex:Current}, let $V=V^k(Cur(\g))$ be the universal enveloping affine vertex algebra of level $k$ endowed with the $\lambda$-bracket $[\, _\lambda \, ]$. Let  $V_\epsilon$ be the regular family of vertex algebras such that $V_\epsilon=V[\epsilon]$,  $[a_\lambda b]_\epsilon=\epsilon[a_\lambda b]$ for $a,b\in Cur(\g)$ and the normally ordered product on $V_\epsilon$ is induced by that on $V$. Then the quasi-classical limit of $V_\epsilon$ is  $S^k(R)$ in Example \ref{Ex:CurrPVA}.
\end{ex}

\begin{rem}
Analogously, we obtain a Poisson superalgebra as the quasi-classical limit of a regular family of associative superalgebras with commutators.
\end{rem}

\begin{defn}
Let $\mathcal{V}$ be a Poisson vertex algebra and let $H:\mathcal{V}\to \mathcal{V}$ be a diagonalizable operator. Denote by $\Delta_a$ the eigenvalue of $H$ corresponding to an eigenvector $a\in \mathcal{V}$. 
If the operator $H$ satisfies that 
\[ \Delta_1=0, \quad \Delta_{\partial a} = 1+\Delta_a, \quad \Delta{a_{(n)}b}=\Delta_a+\Delta_b-n-1\]
for eigenvectors $a,b \in \mathcal{V}$ of $H$ and $n \in \ZZ_{\geq 0}$ then $H$ is called a {\it Hamiltonian operator}. If $H$ is a Hamiltonian operator then the eigenvalue $\Delta_a$ is called the {\it conformal weight} of $a.$
\end{defn}
   
\begin{rem}
A main source of Hamiltonian operator of $\mathcal{V}$ is an {\it energy momentum field} $L\in \mathcal{V}$. Precisely, if $L$ satisfies 
\begin{enumerate}
\item $\{L_\lambda L\}=(\partial+2\lambda)L+\frac{c}{12} \lambda^3$ for the central charge $c\in \CC$,
\item $L_{-1}:=L_{(0)}=\partial$, 
\item $L_{0}:=L_{(1)}$ is a diagonalizable operator on $\mathcal{V}$
\end{enumerate}
then $L_{0}$ is a Hamiltonian operator on $\mathcal{V}.$
\end{rem}

\begin{defn} \cite{DK,Z}
Suppose the Poisson vertex algebra $\mathcal{V}$ has a Hamiltonian operator $H$ and let $J$ be the associative algebra ideal of $\mathcal{V}$ generated by $(\partial+H)a.$ Then the {\it $H$-twisted  Zhu algebra} $Zhu_H(\mathcal{V}):=\mathcal{V}/J$ is the Poisson algebra endowed with the Poisson bracket 
\[ \{\overline{a},\overline{b}\}= \sum_{j\in \ZZ_+} { \Delta_a -1 \choose j } \overline{a_{(j)}b},\quad \text{ for } a,b\in \mathcal{V}.\]
\end{defn}

\begin{ex} \label{Ex:Zhu_Alge}
Let $\g$ be a Lie superalgebra with even $\sll_2$-triple $(e,2x, f)$ and the supersymmetric invariant bilinear form $(\, |\, )$ such that $(e|f)=2(x,x)=1$. Take dual bases $\{u_\alpha| \alpha \in \bar{S}\}$ and $\{u^\alpha|\alpha \in \bar{S}\}$ of $\g$ with respect to the bilinear form $(\, |\, )$.
\begin{enumerate}
\item
 Let $L=\sum_{\alpha\in \bar{S}} \frac{1}{2k} u^\alpha u_\alpha \in S^k(R)$, where $S^k(R)$ is the Poisson vertex algebra in Example \ref{Ex:Current_ClassicalLimit}. Then 
\[ \{L_\lambda L\}= (\partial+2\lambda)L, \quad \{L_\lambda u_\alpha\}=(\partial+\lambda)u_\alpha.\]
Hence $L$ is an energy momentum field and $L_0$ is a Hamiltonian operator of $S^k(R).$ For the Hamiltonian operator $H=L_{0}$, the $H$-twisted Zhu algebra of $S^k(R)$ is the Poisson superalgebra $S(\g)$ endowed with the bracket 
\[ \{a,b\}=[a,b], \quad a,b\in \g.\] 
\item
Let $L_\g=\sum_{\alpha\in \bar{S}} \frac{1}{2k} u^\alpha u_\alpha +\partial x\in S^k(R)$. Then 
\[ \{L_{\g\, \lambda} u_\alpha\}=(\partial+(1-j_\alpha)\lambda)u_\alpha -k \lambda^2(x|u_\alpha)\]
where $[x,u_\alpha]=j_\alpha u_\alpha.$ Moreover, $H:=L_{\g\, (1)}$ is a Hamiltonian operator on $S^k(R)$. The conformal weight $\Delta_\alpha$ of $u_\alpha$ is $1-j_\alpha.$  
The $H$-twisted Zhu algebra of $S^k(R)$ is $S^k(r):=\CC[u_\alpha|\alpha \in \bar{S}]$ endowed with the Poisson bracket 
\[\{ u_\alpha, u_\beta\}=[u_\alpha, u_\beta]-j_\alpha k(u_\alpha|u_\beta), \quad \alpha, \beta\in \bar{S}.\] 
If we denote $v_a=a-k(x,a)$ for $a\in \g$ then $\{v_\alpha, v_\beta \}= v_{[a,b]}.$ Hence the associative superalgebra automorphism $a \mapsto v_a$ of $\CC[u_\alpha|\alpha \in \bar{S}]$ is a Poisson superalgebra isomorphism between $S^0(r)$ and $S^k(r)$. As a conclusion $\CC[u_\alpha |\alpha \in \bar{S}]$ endowed with the Poisson bracket $\{u_\alpha, u_\beta\}=[u_\alpha, u_\beta]$ is the $H$-twisted Zhu algebra of $S^k(R).$
\end{enumerate}
\end{ex}

\begin{rem}
If a vertex algebra $V$ has a Hamiltonian operator $H$, we analogously find an associative superalgebra with commutator which is called the $H$-twisted Zhu algebra $Zhu_H(V)$ of $V$.
\end{rem}

We summarize relations between vertex algebras, Poisson vertex algebras, associative algebras and Poisson algebras by the following diagram:

\[ \xymatrixcolsep{7pc}\xymatrix{
(V, H) \ar[r]^{\text{quasi-classical limit}}  \ar[d]_{Zhu_{H}} & (\mathcal{V},\mathcal{H})  \ar[d]^{Zhu_{\mathcal{H}}} \\
Zhu_H(V) \ar[r]_{\text{quasi-classical limit}}   &  Zhu_{\mathcal{H}}(\mathcal{V}) }. \]
where $V$ is a vertex algebra with a Hamiltonian operator $H$ and $\mathcal{V}$ is a Poisson vertex algebra with a Hamiltonian operator $\mathcal{H}$. 

\subsection{Nonlinear Lie superalgebras and Nonlinear Lie conformal algebras} \label{Subsec:nonlinear}\ \\

In this section, we briefly review constructions of nonlinear Lie superalgebras and nonlinear Lie conformal algebras. We refer to \cite{DK} for details. 

Let $\Gamma_+$ be a discrete additive closed subset of $\RR_+$ containing $0$ and $\Gamma'_+=\Gamma_+\backslash \{0\}$. For $\zeta\in \Gamma'_+$, we denote by $\zeta_-$ the largest element of $\Gamma_+$ strictly smaller than $\zeta.$

Let $\g$ be a $\Gamma'_+$-graded vector superspace and $\mathcal{T}(\g)$ be the tensor superalgebra over $\g$.
Denote by $\zeta(a)$ the $\Gamma'_+$-grading of $a\in \g$.   Then $\mathcal{T}(\g)$ is a $\Gamma_+$-graded algebra 
\[ \mathcal{T}(\g)= \bigoplus_{\zeta\in \Gamma_+} \mathcal{T}(\g)[\zeta]\]
induced by the $\Gamma'_+$-grading of $\g$. More precisely,  $\zeta(c)=0$ for $c\in \CC$ and $\zeta(A\otimes B)= \zeta(A)+\zeta(B)$ for $A, B\in \mathcal{T}(\g)$.  Then there is an increasing filtration of $\mathcal{T}(\g)$
\[ \mathcal{T}_\zeta(\g)= \bigoplus_{\zeta'\leq \zeta}\mathcal{T}(\g)[\zeta'].\]

If $\g$ is endowed with the linear map 
\[ [\, , \, ]: \g \otimes \g \to \mathcal{T}(\g)\]
then we can extend the bracket $[\, , \, ]$ defined on $\g$ to the bracket $[\, , \, ]$ defined on $\mathcal{T}(\g)$ by Leibniz rules. (See \cite{DK}.)

\begin{defn} \cite{DK} \label{Def:nonlinearLA}
If $\g$ is endowed with the linear map 
\[ [\, , \, ]: \g \otimes \g \to \mathcal{T}(\g)\]
such that 
\begin{equation*}
\begin{aligned}
&\text{(grading condition) $[a, b]\in \mathcal{T}_{(\zeta(a)+\zeta(b))_-}(\g)$};\\
&\text{(skewsymmetry) $[a,b]=-(-1)^{p(a)p(b)}[b,a]$} ;\\
&\text{(Jacobi identity) $[a,[b,c]]-(-1)^{p(a)p(b)}[b,[a,c]]-[[a,b],c] \in \mathcal{M}_{(\zeta(a)+\zeta(b)+\zeta(c))_-}(\g)$};\\
\end{aligned}
\end{equation*}
where  $a,b,c\in \g$ and  $\mathcal{M}_\zeta(\g)=\mathcal{M}(\g)\cap\mathcal{T}_\zeta(\g)$ for
\[\mathcal{M}=span\{A\otimes (d\otimes e -(-1)^{p(e)p(d)} e\otimes d -[d,e])\otimes D|d,e\in \g, A,D\in \mathcal{T}(\g)\}\]
then $\g$ is called a {\it nonlinear Lie superalgebra}. 
\end{defn}

\begin{defn}\cite{DK}
Let $\g$ be a $\Gamma'_+$-graded nonlinear Lie superalgebra. The associative algebra $U(\g)=\mathcal{T}(\g)/ \mathcal{M}(\g)$ is called the {\it universal enveloping algebra } of $\g.$
\end{defn}

Consider the quasi-classical limit of the regular family of Lie superalgebras $U_\epsilon:=U(\g)[\epsilon]$, which is endowed with the bracket defined by $[a,b]_\epsilon=\epsilon[a,b]$ for $a,b\in \g$, is the Poisson  superalgebra $S(\g)$ endowed with the Poisson bracket $\{\, , \, \}$ defined by $\{a,b\}=[a,b]$ for any $a,b\in \g$ and Leibniz rules. 

\begin{ex}
Let $\g$ be a $\ZZ_{>0}$-graded vector superspace such that $\zeta(a)=1$ for all $a\in \g.$ If $\g$ is a Lie superalgebra endowed with the Lie bracket $[\, , \, ]$ and the supersymmetric bilinear invariant form $(\, |\, )$ then the linear map $[\, , \, ]_k:\g\otimes \g \to \mathcal{T}(\g)$ such that 
\[ [a,b]_k=[a,b]+\frac{k}{2}(h|[a,b]) \qquad \text{ for } k\in \CC, \quad h\in \g_{\bar{0}} \] 
is a nonlinear Lie bracket of $\g.$ 
Consider the quasi-classical limit of the regular family of Lie superalgebras $U_\epsilon:=U(\g)[\epsilon]$, which is endowed with the bracket defined by $[a,b]_\epsilon=\epsilon[a,b]$ for $a,b\in \g$, is the Poisson  superalgebra $S(\g)$ endowed with the Poisson bracket $\{\, , \, \}$ defined by $\{a,b\}=[a,b]$ for any $a,b\in \g$ and Leibniz rules. 

\end{ex}

\vskip 3mm

Analogously, we can define nonlinear Lie conformal algebras and their universal enveloping vertex algebras. Here we briefly review the definition. (See \cite{DK}.)

Let $R$ be a $\CC[\partial]$-module with $\Gamma'_+$ grading and $\mathcal{T}(R)$ be the tensor superalgebra over $R$. Denote by $\zeta(a)$ the $\Gamma'_+$-grading of $a\in R$.   Then $\mathcal{T}(R)$ is a $\Gamma_+$-graded algebra 
\[ \mathcal{T}(R)= \bigoplus_{\zeta\in \Gamma_+} \mathcal{T}(R)[\zeta]\]
induced by the $\Gamma'_+$-grading of $R$. More precisely,  $\zeta(c)=0$ for $c\in \CC$ and $\zeta(A\otimes B)= \zeta(A)+\zeta(B)$ for $A, B\in \mathcal{T}(R)$.  Then there is an increasing filtration of $\mathcal{T}(R)$
\[ \mathcal{T}_\zeta(R)= \bigoplus_{\zeta'\leq \zeta}\mathcal{T}(R)[\zeta'].\]

If $R$ is endowed with the linear map 
\[ [\, _\lambda, \, ]: \g \otimes \g \to \CC[\lambda]\otimes \mathcal{T}(R)\]
satisfying 
\begin{equation*}
\begin{aligned}
&\text{ (grading condition) $ [a_\lambda b] \in \CC[\lambda]\otimes \mathcal{T}_{(\zeta(a)+\zeta(b))_-}(R)$};\\
&\text{ (sesquilinearity) $[a_\lambda \partial b]=(\lambda+\partial)[a_\lambda b], \quad [\partial a_\lambda b]=-\lambda[a_\lambda b]$}
\end{aligned}
\end{equation*}
for $a,b,c\in R$ then we can extend the $\lambda$-bracket on $R$ to that on $\mathcal{T}(R)$ via Definition \ref{Def:VA}. Precise construction for the normally ordered product and $\lambda$-brackets on $\mathcal{T}(R)$ can be found in \cite{DK}.

\begin{defn} \cite{DK}
 Let $R$ be a $\CC[\partial]$-module endowed with the $\lambda$-bracket $[\, _\lambda\, ]:R\otimes R \to \CC[\lambda]\otimes \mathcal{T}(R)$ with grading conditions and sesquilinearities. If the $\lambda$-bracket satisfies
 \begin{equation*}
\begin{aligned} 
& \text{ (skewsymmetry)  $[a_\lambda b]=-(-1)^{p(a)p(b)}[b_{-\partial-\lambda} a]$ } ;\\
& \text{ (Jacobi identity)  $[a_\lambda[b_\mu c]]-(-1)^{p(a)p(b)}[b_\mu[a_\lambda c]] -[[a_\lambda b]_{\lambda+\mu} c] \in \mathcal{M}_\zeta(R)$};
\end{aligned}
\end{equation*}
for $a,b,c\in R$ and $\mathcal{M}_\zeta(R)=\mathcal{M}(R) \cap \mathcal{T}_\zeta(R)$ where the subset  $\mathcal{M}(R)\subset \mathcal{T}(R)$ is 
\[  \left\{ A\otimes \left. \left(d\otimes e\otimes D -(-1)^{p(d)p(e)} e\otimes d \otimes D - : \int_{-\partial}^0 [d_\lambda e] d\lambda \,  D:\right)\, \right| \, d,e\in R, \, A, D\in \mathcal{T}(R)\right\}\]
then $R$ is called a {\it nonlinear Lie conformal algebra.}
\end{defn}

\begin{defn} \cite{DK}
Let $R$ be a $\Gamma'_+$-graded nonlinear Lie conformal algebra. The vertex algebra $\mathcal{T}(R)/\mathcal{M}(R)$ is called the {\it universal enveloping vertex algebra} of $R.$ 
\end{defn}

\begin{ex}
Let $\g$ be a Lie superalgebra endowed with the Lie bracket $[\, , \, ]$ and the supersymmetric invariant bilinear form $(\, |\, )$ and let $R=\CC[\partial]\otimes \g$ be a Lie conformal algebra endowed with the $\lambda$-bracket such that $[a_\lambda b]=[a,b]$ for $a,b\in \g$. Consider the $\ZZ_{>0}$-grading on $R$ defined by $\zeta(a)=1$ for any $a\in R.$ Then the map $[\, _\lambda \, ]_k:R \otimes R \to \CC[\lambda]\otimes \mathcal{T}(R)$ such that 
\[ [a_\lambda b]_k= [a,b]+k\lambda(a|b) , \qquad k \in \CC \]
is a nonlinear $\lambda$-bracket of $R.$ Hence $(R, [\, _\lambda\, ])$ is a nonlinear LCA. Consider the quasi-classical limit of the regular family of vertex algebras $V_\epsilon:=V(R)[\epsilon]$, which is endowed with the $\lambda$ bracket defined by $[a_\lambda b]_\epsilon=\epsilon[a_\lambda b]$ for $a,b\in R$, is the Poisson  vertex algebra $S(R)$ endowed with the Poisson $\lambda$-bracket $\{\, _\lambda \, \}$ defined by $\{a_\lambda b\}=[a_\lambda b]$ for any $a,b\in R$ and Leibniz rules. 
\end{ex}

\subsection{ Basic results in filtered complexes } \label{Subsec:FC} \ \\

Let $\Gamma= \frac{1}{N} \ZZ$ for a positive integer $N$ and $U$ be a vector superspace.  The linear map $d:U\to U$ is called  an odd differential of  $U$, if $d^2=0$ and $d$ is odd. If $U$ is a (i) superalgebra,  (ii) Lie superalgebra,  (iii) Lie conformal algebra, respectively, then we assume that 
\[ (i) \ d(ab)= d(a) b+ (-1)^{p(a)} a d(b), \ (ii) \ d([a,b])= [d(a), b]+ (-1)^{p(a)} [a,d(b)], \]
\[ \ (iii) \ d[a_\lambda b]= [d(a)_\lambda b] +(-1)^{p(a)}[a_\lambda d(b)],\]
respectively. 

\begin{defn} \cite{DK}
The complex $(U,d)$ is called a {\it filtered complex}  if 
\begin{enumerate}
\item $U$ is a $\Gamma$-bigraded space such that 
\[ U= \bigoplus_{p,q \in \Gamma} U^{p,q}= \bigoplus_{p+q=n \in \ZZ} U^n\]
\item For the decreasing filtration $\{ F^p U \, | \,  p\in \Gamma \}$ where $F^p U=\bigoplus_{p' \geq p, q} U^{p', q} $,  the odd differential $d$ has degree $1$ and preserves the filtration:
\[ d(F^p U^n ) \subset F^p U^{n+1}\]
where $F^p U^n = F^p U  \cap U^n.$  
\end{enumerate}
\end{defn}
  
If the complex $(U,d)$ is a filtered complex with the filtration $\{F^p U\, | \, p\in \Gamma\}$ then $H^n(U,d)$ is also a filtered space with 
\[ F^p H^n (U,d)= \text{ Ker}(d |_{F^pU^n}) / ( \text{Im} d \cap F^p U^n ).\]

Let us write
\[ \text{gr}^{pq} H(U,d)= F^p H^{p+q}(U,d) / F^{p+\epsilon} H^{p+q}(U,d) \]
for $\epsilon= \frac{1}{N}$ and let 
\[ H^{p,q}(\text{gr} U, d^{\text{p,q}}) = \frac{ \text{Ker}(d^{\text{gr}}: \text{gr}^{p,q} U \to \text{gr}^{p,q+1} U)}{ \text{Im}(d^{\text{gr}}: \text{gr}^{p,q-1} U \to \text{gr}^{p,q} U)} \]
where $\text{gr} U= \bigoplus_{p,q\in \Gamma} \text{gr}^{p,q} U$ and $\text{gr}^{p,q} U = F^p U^{p+q} /F^{p+\epsilon} U^{p+q}.$

\begin{defn} \cite{DK}
Let $(U,d)$ be a filtered complex. 
\begin{enumerate}
\item
The complex $(U,d)$ is said to be {\it good } if $H^{p,q}(\text{gr}U, d^{\text{gr}})=0$ for all $p,q \in \Gamma$ such that $p+q \neq 0.$
\item 
For each $n>0$, if $F^p U^n=0$ for $p>>0$ then $U$ is said to be {\it locally finite.}
\end{enumerate}
\end{defn}

\begin{prop} \label{Prop:com_gr}\cite{DK} 
If the filtered complex $(U,d)$ is good and locally finite then we have 
\[ \gr^{p,q}H(U,d) \simeq H^{p,q} (\gr \, U, d^\gr). \]
\end{prop}

\begin{prop} [K\"unneth lemma]
\begin{enumerate}
\item
Let $V_1$ and $V_2$ be vector superspaces with differentials $d_i: V_i \to V_i$ for $i=1,2$. If $d:V\to V$ is a differential on $V=V_1 \otimes V_2$ such that $d=d_1 \otimes 1 +1 \otimes d_2$ then there is a canonical linear isomorphism
\[ H(V,d)= H(V_1, d_1) \otimes H(V_2, d_2).\]
\item If $S(V)$ is a supersymmetric algebra generated by the vector superspace $V$ then 
\[ H(S(V), d) \simeq S(H(V,d)).\]
\end{enumerate}
\end{prop}

\begin{rem} \cite{DK}
Let $\g$ be a $\Gamma'_+$-graded nonlinear Lie superalgebra with a differential $d:\g\to \g$ preserving the $\Gamma_+'$-grading. Suppose (1) $H(\g,d)$ has the $\Gamma'_+$-grading induced from that of $\g$, (2) nonlinear Lie bracket of $\g$ induces a nonlinear lie  bracket of $H(\g,d)$. Then there is a canonical associative superalgebra isomorphism 
\[ H (U(\g), d)\simeq U(H(\g,d)).\]
\end{rem}

\section{Definition of classical affine W-algebras associated to Lie superalgebras} \label{Sec:superW}

Let $\g=\g_{\bar{0}} \oplus \g_{\bar{1}}$ be a classical finite simple Lie superalgebra with the even part $\g_{\bar{0}}$ and the odd part $\g_{\bar{1}}.$ We choose an even $\sll_2$-triple $(e,h=2x,f)$ in $\g_{\bar{0}}$. Then the operator $\ad x$ on $\g$ is diagonalizable and 
\[ \textstyle \g=\bigoplus_{i\in \frac{\ZZ}{2}} \g(i) \text{ where } \g(i)=\{ g\in \g\, |\, [x,g]=ig\}.\]
 Especially, $f\in \g(-1)$ and $e\in \g(1).$  Also, let $(\, | \, )$ be a supersymmetric bilinear invariant form which satisfies $(e|f)=\frac{1}{2}(h|h)=1$ and let  
\[ \textstyle \n =\bigoplus_{i>0} \g(i), \quad \n_-=\bigoplus_{i<0} \g(i), \quad \m=\bigoplus_{i\geq 1} \g(i)\]
 be subalgebras of $\g.$ The following two sets   
\[ \{u_\alpha| \alpha \in \overline{S}\} \text{ and } \{u^\alpha|\alpha \in \overline{S}\} \] 
are dual bases of $\g$ such that  (1) both of bases are compatible with the parity, (2) $(\, u_\alpha\, | \, u^\beta\, )= \delta_{\alpha\beta}$, (3) $u_\alpha \in \g(j_\alpha)$ and $u^\alpha \in \g(-j_\alpha)$ . Let $S$ be the subset of $\overline{S}$ such that 
\[ \{u_\alpha| \alpha \in S\} \text{ and } \{u^\alpha|\alpha \in S\} \] 
be dual bases of $\n$ and $\n_-$.  The subset $S(1/2)\subset S$ is the index set such that $\{u_\alpha | \alpha \in S(1/2)\} = \{u_\alpha | \alpha \in S\}\cap \g\left(\frac{1}{2}\right)$. Thus 
\[ \textstyle \{u_\alpha|\alpha\in S(1/2)\} \text{  is a basis of } \g\left(\frac{1}{2}\right).\]

\subsection{First definition via  classical BRST complex} \label{Subsec:superW-BRST}\ \\ 

Recall that a quantum W-algebra is defined by a BRST quantized complex of a complex of Lie (super)algebras. We shall call by the {\it classical BRST complex}, the quasi-classical limit of the BRST quantized complex.

In order to introduce a classical BRST complex, we recall following three types of nonlinear Lie conformal algebras \cite{DK}:

\begin{enumerate}
\item The nonlinear current Lie conformal algebra $Cur_k(\g) =\CC[\partial]\otimes \g$ is endowed with the $\lambda$-bracket 
\[ [a_\lambda b]= [a,b]+ k \lambda(a|b) , \qquad a,b\in \g \]
for given $k\in \CC$.

\item Let $\phi_\n$ be a vector superspace isomorphic to $\Pi(\n)$ where $\Pi$ is the parity reversing map and, similarly, let $\phi^{\n_-}\simeq \Pi(\n_-)$ as vector superspaces. Then the 
charged free fermion nonlinear Lie conformal algebra $R_{ch}=\CC[\partial] \otimes  (\phi_\n \oplus \phi^{\n^-})$ is endowed with the $\lambda$-bracket 
\[ [ \phi_{a_1}\, _\lambda\, \phi_{a_2}] = [ \phi^{b_1} \, _\lambda\, \phi^{b_2} ]=0, \quad [ \phi_a \, _\lambda\, \phi^b]= (a|b), \]
for $a_1, a_2, a \in \n$ and $b_1, b_2, b \in \n_-.$ For $a\in \g$, we let $\phi_a= \phi_{\pi_+ a}$ and $\phi_a= \phi_{\pi_- a}$, where $\pi_+$ and $\pi_-$ are projection maps from $\g$ onto $\n$ and $\n_-.$

\item Let $\Phi_{\g\left(\frac{1}{2}\right)}$ be a vector superspace isomorphic to $\g\left(\frac{1}{2}\right)$. The neutral free fermion nonlinear Lie conformal algebra $R_{ne}=\CC[\partial]\otimes \Phi_{\g\left(\frac{1}{2}\right)}$ is endowed with the $\lambda$-bracket 
\[ [ \Phi_{c_1} \, _\lambda \, \Phi_{c_2}] = (f|[c_1, c_2]). \]
For $a \in \g$, we let $\Phi_a= \Phi_{\pi_{1/2} a},$ where $\pi_{1/2}$ is the projection map on $\g$ onto $\g\left(\frac{1}{2}\right).$ 

\end{enumerate}

Let $R=Cur_k(\g)\oplus R_{ch}\oplus R_{ne}$ be the direct sum of $Cur_k(\g)$, $R_{ch}$ and $R_{ne}$ as a nonlinear LCA. The supersymmetric algebra $S(R)$ generated by  $R$ is a PVA endowed with the bracket $\{\, _\lambda\, \}$ induced by that of $R$ and Leibniz rules. 

Denote
\begin{enumerate} 
\item $L_\g= \sum_{\alpha\in \bar{S}} \frac{1}{2k} u^\alpha u_\alpha  +\partial x \in S(Cur_k(\g))$ where $\{ u^\alpha |\, \alpha\in \bar{S}\}$ and $\{u_\alpha |\,\alpha \in \bar{S}\}$ ;
\item $L^{ch}=- \sum_{\alpha \in S} j_\alpha \phi^\alpha (\partial \phi_{\alpha})+\sum_{\alpha \in S} (1-j_\alpha)(\partial \phi^\alpha)\phi_\alpha\in S(R_{ch})$ where $\phi_\alpha:=\phi_{u_\alpha}$ and $\phi^\alpha:=\phi^{u^\alpha}$;
\item $L^{ne}= \frac{1}{2} \sum_{\alpha \in S(1/2)} (\partial \Phi^\alpha) \Phi_\alpha\in S(R_{ne})$ where $\Phi^\alpha=\Phi_{v_\alpha}$ and $\Phi_\alpha=\Phi_{u_\alpha}$ such that  $(f|[u_\alpha, v^\beta])=\delta_{\alpha\beta}$ for $\alpha, \beta \in S(1/2).$
\end{enumerate}

Then 
\begin{equation}
\begin{aligned}
&  \{ L^\g \, _\lambda u_\alpha \}= \partial u_\alpha + (1-j_\alpha)\lambda u_\alpha-k \lambda^2 (x|u_\alpha)  \text{ for } \alpha \in \bar{S} \text{ and } u_\alpha \in \g(j_\alpha); \\
&  \{L^{ch}\, _\lambda \phi_\alpha\}= (\partial +(1-j_\alpha)\lambda)\phi_\alpha, \quad  \{L^{ch}\, _\lambda \phi^\alpha\}= (\partial +j_\alpha \lambda)\phi^\alpha \text{ for } \alpha \in S; \\
& \{L^{ne}\, _\lambda \Phi_\alpha\}=  (\partial +\frac{1}{2}\lambda)\Phi_\alpha \quad \text{ for } \alpha \in S(1/2).
\end{aligned}
\end{equation}
Hence 
\begin{equation} \label{Eqn:Ham}
H=L_{(1)},  \text{ where }   L=L^\g+L^{ch}+L^{ne} \in S(R),
\end{equation}
is a Hamiltonian operator of $S(R)$ and conformal weights of generating elements of $R$ are
\[ \Delta_{u_\alpha} = 1- j_\alpha, \quad \Delta_{\phi^\beta}=j_\beta, \quad \Delta_{\phi_\beta}=(1-j_\beta), \quad\Delta_{\Phi_\gamma} =\frac{1}{2}\]
for $\alpha \in \bar{S}$, $\beta\in S$ and $\gamma \in S(1/2).$

Take the element
\begin{equation} \label{differential}
d= \sum_{\alpha \in S} (-1)^{p(\alpha)} \phi^\alpha u_\alpha + \sum_{a\in S(1/2)}\phi^\alpha \Phi_\alpha + \phi^f +\frac{1}{2} \sum_{\alpha, \beta} (-1)^{p(\alpha)} \phi^\alpha \phi^\beta \phi_{[u_\beta, u_\alpha]} \in S(R),
\end{equation}
where $\phi_\alpha=\phi_{u_\alpha}$, $\phi^\alpha= \phi^{u^\alpha}$ and $p(\alpha)=p(u_\alpha).$
Then we have the following lemma.

\begin{lem}
\begin{enumerate}
\item
The element $d \in S(R)$ has the odd parity.
\item
We have the following formulas:
\begin{equation}
\begin{aligned}
&\{ d\, _\lambda a\}  = \sum_{\alpha \in S} (-1)^{p(\alpha)} \phi^\alpha[u_\alpha, a] + k\, (-1)^{p(a)}(\partial+\lambda)\phi^a, \\
&\{ d\, _\lambda \phi_a\}  = \pi_+ a + (a|f)+ (-1)^{p(a)} \Phi(a)+ \sum_{\alpha \in S} \phi^\alpha \phi_{[u_\alpha, \pi_+ a]}, \\
&\{ d\, _\lambda \phi^a\} = \frac{1}{2} \sum_{\alpha \in S} (-1)^{p(\alpha)} \phi^\alpha \phi^{[u_\alpha, a]}, \\
&\{d\, _\lambda \Phi_a \} = \phi^{[\pi_{1/2} a, f]}.
\end{aligned}
\end{equation}
\item
We have $\{d_\lambda d\}=0.$
\end{enumerate}
\end{lem}

\begin{proof}
 Let us denote $s(a)=(-1)^{p(a)}$ for a homogenous element $a\in S(R)$. 

(1)  Since $s(\phi^\alpha)s(u_\alpha)=-1$,  $s(\phi^\beta) s(\Phi_\beta)=-1$, and $s(\phi^f)=-1$, for $\alpha \in S$,  $\beta \in S(1/2)$, the element $\sum_{\alpha \in S} p(\alpha) \phi^\alpha u_\alpha + \sum_{a\in S(1/2)}\phi^\alpha \Phi_\alpha + \phi^f $ has the odd parity. Also, we have $s(\phi^\alpha \phi^\beta \phi_{[u_\beta, u_\alpha]})= s(\phi^\alpha)s(\phi^\beta)s(\phi_{[u_\beta, u_\alpha]})=-1.$ Hence $d$ is an odd element. 

\vskip 3mm 

(2) Observe that 
\begin{equation*}
\{d\, _\lambda a\} = \sum_{\alpha\in S} \left( s(u_\alpha) \phi^\alpha[u_\alpha, a]  + k\,  s(u_\alpha) (\partial+\lambda)(u_\alpha | a) \phi^\alpha \right).
\end{equation*}
Let $a$ be a homogeneous element. Then $(u_\alpha|a)\neq 0$ only if $p(\alpha)=p(a).$ Hence 
\[ \{d\, _\lambda a\}  = \sum_{\alpha \in S} s(u_\alpha) \phi^\alpha[u_\alpha, a] + k\, s(a)(\partial+\lambda)\phi^a.\]
If we write $X_a=s(a)a+\Phi_a + (a|f)$ for $a\in \n$ and $X_\alpha= X_{u_\alpha}$ for $\alpha\in S$ then
\begin{equation*}
\begin{aligned}
\{ d\, _\lambda \phi_a\} & =   \sum_{\alpha\in S}  \{ \phi^\alpha X_\alpha \, _\lambda \phi_a\} +  \sum_{\alpha, \beta\in S} \frac{1}{2} \{ s(u_\alpha) \phi^\alpha \phi^\beta \phi_{[u_\beta, u_\alpha]} \, _\lambda \phi_a\}.
\end{aligned}
\end{equation*}
We have 
\begin{equation} \label{Eqn:4.2_150205}
\sum_{\alpha\in S}  \{ \phi^\alpha X_\alpha \, _\lambda \phi_a\} = \sum_{\alpha\in S} \phi(u_\alpha) X_\alpha (a|u^\alpha)=\pi_+a +(a|f) + s(a) \Phi(a)
\end{equation}
and 
\begin{equation} \label{Eqn:phi_a}
\begin{aligned}
 & \sum_{\alpha, \beta \in S}  \{ s(u_\alpha) \phi^\alpha \phi^\beta \phi_{[u_\beta, u_\alpha]} \, _\lambda \phi_a\} \\
  & = \sum_{\alpha, \beta \in S} s(u_\alpha) \phi_{[u_\beta, u_\alpha]} (-1)^{p(\phi^\alpha) p(\phi^\beta)} \phi^\beta \{\phi^\alpha \, _\lambda \phi_a\} + \sum_{\alpha, \beta\in S} s(u_\alpha) \phi_{[u_\beta, u_\alpha]} \phi^\alpha \{ \phi^\beta\, _\lambda \phi_a\}.
\end{aligned}
\end{equation}
The first term in the RHS of (\ref{Eqn:phi_a}) is 
\begin{equation}
\begin{aligned}
& \sum_{\alpha, \beta \in S} s(u_\alpha) \phi_{[u_\beta, u_\alpha]} (-1)^{p(\phi^\alpha) p(\phi^\beta)} \phi^\beta \{\phi^\alpha \, _\lambda \phi_a\}  \\
& =\sum_{\alpha, \beta\in S} s(u_\alpha) \phi_{[u_\beta, u_\alpha]} (-1)^{(p(\alpha)+1)(p(\beta)+1)} \phi^\beta \{ \phi^\alpha \, _\lambda \phi_a\} = \sum_{\beta\in S}\phi^\beta\phi_{[u_\beta, \pi_+ a]}
\end{aligned}
\end{equation}
and the second term in the RHS of (\ref{Eqn:phi_a}) is 
\begin{equation}
\begin{aligned}
& \sum_{\alpha, \beta\in S} s(u_\alpha) \phi_{[u_\beta, u_\alpha]} \phi^\alpha \{ \phi^\beta\, _\lambda \phi_a\}=\sum_{\alpha\in S} s(u_\alpha)s(a)\phi_{[\pi_+ a, u_\alpha]} \phi^\alpha = \sum_{\alpha\in S} \phi^\alpha \phi_{[u_\alpha, \pi_+ a]}.
\end{aligned}
\end{equation}
By (\ref{Eqn:4.2_150205}) and (\ref{Eqn:phi_a}), we have  $\{ d\, _\lambda \phi_a\}  = \pi_+ a + (a|f)+ s(a) \Phi(a)+ \sum_{\alpha \in S} \phi^\alpha \phi_{[u_\alpha, \pi_+ a]}.$ \\
The rest of two equations in (2) can be obtained by similar computations.

\vskip 3mm

(3) The element $d= \sum_{\alpha\in S} \phi^\alpha X_\alpha+\frac{1}{2}\sum_{\alpha, \beta\in S} s(u_\alpha) \phi^\alpha \phi^\beta \phi_{[u_\beta, u_\alpha]}$. By direct computations, we have
\begin{equation} \label{Eqn:4.7_150207}
 \big\{ \ \sum_{\alpha\in S} \phi^\alpha X_\alpha \, _\lambda \sum_{\beta \in S} \phi^\beta X_\beta \ \big\} = -\sum_{\alpha, \beta \in S} s(u_\alpha)\phi^\alpha \phi^\beta X_{[u_\beta, u_\alpha]} 
\end{equation}
and
\begin{equation}\label{Eqn:4.8_150207}
\begin{aligned}
& \big\{\ \sum_{\alpha\in S} \phi^\alpha X_\alpha \, _\lambda \sum_{\gamma, \delta \in S} \frac{1}{2} s(u_\gamma) \phi^\gamma \phi^\delta \phi{[u_\delta, u_\gamma]}\ \big\}=\sum_{\gamma, \delta \in S} \frac{1}{2} s(u_\gamma)\phi^\gamma \phi^\delta X_{[u_\delta, u_\gamma]}. \\
\end{aligned}
\end{equation}
On the other hand, we have
\begin{equation} \label{Eqn:4.9_150207}
\begin{aligned}
&\left\{ \,\sum_{\alpha, \beta\in S}  \frac{1}{2} s(u_\alpha) \phi^\alpha \phi^\beta \phi{[u_\beta, u_\alpha]} \right.  \, _\lambda \left. \sum_{\gamma, \delta\in S} \frac{1}{2} s(u_\gamma) \phi^\gamma \phi^\delta \phi{[u_\delta, u_\gamma]} \, \right\} \\
& \hskip 10mm = \sum_{\alpha, \beta, \gamma, \delta \in S} \frac{1}{4} s(u_\alpha) s(u_\gamma)   \left(\phi^\alpha \phi^\beta \{ \phi_{[u_\beta, u_\alpha]} \, _\lambda \phi^\gamma \phi^\delta \} \, \phi_{[u_\delta, u_\gamma]}  \right. \\
& \hskip 60mm \left. +  \phi_{[u_\beta, u_\alpha]} \{ \phi^\alpha \phi^\beta \, _\lambda \phi_{[u_\delta, u_\gamma]} \} \phi^\gamma \phi^\delta \right).
\end{aligned}
\end{equation}
By Leibniz rule,
\begin{equation*}
\begin{aligned}
& \sum_{\alpha, \beta, \gamma, \delta \in S}  s(u_\alpha) s(u_\gamma) \phi^\alpha \phi^\beta \{ \phi_{[u_\beta, u_\alpha]}\, _\lambda \phi^\gamma \phi^\delta \} \, \phi_{[u_\delta, u_\gamma]} \\
& = -\sum_{\alpha, \beta, \delta \in S} s(u_\alpha) \phi^\delta \phi^\alpha \phi^\beta \phi_{[[u_\beta, u_\alpha], u_\delta]}- \sum_{\alpha, \beta, \gamma \in S}  s(u_\alpha) \phi^\gamma \phi^\alpha \phi^\beta \phi_{[[u_\beta, u_\alpha], u_\gamma]}\\
 & = -2 \sum_{\alpha, \beta \delta \in S} s(u_\alpha) \phi^\delta \phi^\alpha \phi^\beta \phi_{[[u_\beta, u_\alpha], u_\delta]}
\end{aligned}
\end{equation*}
and 
\begin{equation*}
\begin{aligned}
& \sum_{\alpha, \beta, \gamma, \delta \in S}  s(u_\alpha) s(u_\gamma)\phi_{[u_\beta, u_\alpha]} \{ \phi^\alpha \phi^\beta \, _\lambda \phi_{[u_\delta, u_\gamma]} \} \phi^\gamma \phi^\delta \\
& = \sum_{\beta, \gamma, \delta \in S} -s(u_\gamma) \phi^\beta \phi^\gamma \phi^\delta \phi_{[[u_\delta, u_\gamma], u_\beta]} + \sum_{\alpha, \gamma, \delta} s(u_\gamma) \phi^\alpha \phi^\gamma \phi^\delta \phi_{[[u_\delta, u_\gamma], u_\alpha]} =0.
\end{aligned}
\end{equation*}
Since $\{d_\lambda d\}= (\ref{Eqn:4.7_150207}) + 2\cdot (\ref{Eqn:4.8_150207})+(\ref{Eqn:4.9_150207})= -\frac{1}{2} \sum_{\alpha, \beta, \delta \in S} s(u_\alpha) \phi^\delta \phi^\alpha \phi^\beta \phi_{[[u_\beta, u_\alpha], u_\delta]}$, we want to show that 
\begin{equation} \label{Eqn:4.10_150207}
 \sum_{\alpha, \beta, \gamma \in S} s(u_\beta) \phi^\alpha \phi^\beta \phi^\gamma \phi_{[[u_\gamma, u_\beta], u_\alpha]}=0.
 \end{equation}
We obtain (\ref{Eqn:4.10_150207}) from the property that, for any $\alpha, \beta, \gamma \in S$, the following formula holds:
\[ s(u_\beta) \phi^\alpha \phi^\beta \phi^\gamma \phi_{[u_\gamma,[u_\beta, u_\alpha]]} 
+ s(u_\alpha) \phi^\gamma \phi^\alpha \phi^\beta \phi_{[u_\beta,[u_\alpha, u_\gamma]]} 
+s(u_\gamma) \phi^\beta \phi^\gamma \phi^\alpha \phi_{[u_\alpha,[u_\gamma, u_\beta]]}=0 .\]
\end{proof}

\begin{prop}
Let $d_{(0)}: S(R) \to S(R)$ be defined by $A \mapsto \{d\, _\lambda A\}|_{\lambda=0}.$ Then $d_{(0)}^2 =0$ and $d_{(0)}$ is a differential on $S(R)$.
\end{prop}

\begin{proof}
By the Leibniz rule, we have 
\[ \{d_\lambda \{ d_\mu A\}\} + \{ d_\mu\{d_\lambda A\}\} = \{\{ d_\lambda d\}_{\lambda+\mu} A\}. \]
If we take $\lambda=\mu=0$ then $d_{(0)}^2 A = \{\{ d_\lambda d\}_{\lambda +\mu} A\}|_{\lambda=\mu=0}$. Since $\{d_\lambda d\}=0$, we have $d_{(0)}^2=0$ and $d_{(0)}$ is a differential on $S(R).$
\end{proof}

\begin{defn}
The classical BRST complex associated to $\g$ and $f$ be $S(R)$ with the differential $d_{(0)}.$ The {\it classical affine W-algebra }  
\[ \WW_1(\g, f,k)=H(S(R), d_{(0)})\]
associated to $\g$ and $f$  is a PVA endowed with the supersymmetric product and the $\lambda$-bracket 
\[ (A+I)(B+I)=AB+I, \quad \{A+I \, _\lambda B+I \}= \{A_\lambda B\}+I, \quad \text{ for } I=\text{Im}\, d_{(0)} \subset S(R). \]
\end{defn}

\begin{note}
If we want to emphasis the W-algebra $\WW(\g,f,k)$ is associated to a Lie ``super''algebra $\g$, we call the algebra by {\it W-superalgebra}. 
\end{note}

In order to show the well definedness of $\WW$-algebras, we need the following proposition.

\begin{prop}
The product and the $\lambda$-bracket on the $\WW_1(\g,f,k)$ are well-defined.
\end{prop}

\begin{proof}
Let $d_{(0)}$ be the differential of the classical BRST complex associated to $\g$ and $f$. By the Leibniz rule and the Jacobi identity, we have 
\[ d_{(0)} (AB) =0, \quad d_{(0)}(\{A_\lambda B\}) =0 \text{ if } A,B\in \text{ker} d_{(0)}.\]
Also, if $A,B\in \text{ker} d_{(0)}$ and $X,Y \in S(R)$ then
\[ ((A+d_{(0)}X)(B+d_{(0)}Y))= AB+ d_{(0)}(s(A)AY+XB+X d_{(0)}Y)\]
and
\[ \{ A+d_{(0)}X\, _\lambda\, B+d_{(0)}Y\}= \{A_\lambda B\}+ d_{(0)}(\{s(A)A_\lambda Y\}+\{X_\lambda B\}+\{X_\lambda d_{(0)}Y\}).\]
Hence $\WW_1(\g,f,k)$ is a PVA. 
\end{proof}

\subsection{Second definition via Hamiltonian reduction}\ \\

Let $S(\CC[\partial]\otimes \g)$ be the supersymmetric algebra generated by the vector superspace $\CC[\partial]\otimes \g$. Take the associative superalgebra ideal 
\[I=MS(\CC[\partial]\otimes \g)\text{ where }  M=\{ m+\chi(m)|m\in \m, \, \chi(m)= (f|m)\}.\] 
Let
\[ \mathcal{V}(\g,f,k)= S(\CC[\partial]\otimes \g)/I\]
be the supersymmetric algebra.
Define the $\ad_\lambda \n$-action on $\mathcal{V}(\g,f,k)$ by 
\[ \ad_\lambda n\, (A) = \{ n_\lambda A\} +I[\lambda] \qquad \text{ for } n\in \n\]
where the bracket $\{ n_\lambda A\}$ is induced from the bracket of $S(Cur_k(\g)).$
Then $\ad_\lambda \n (I) \subset I[\lambda]$ and the subspace
\begin{equation} \label{Eqn:classicalW}
\WW_2(\g,f,k) = \mathcal{V}(\g,f,k)^{\ad_\lambda \n}=\{ A \in \mathcal{V}(\g,f,k) | \, \ad_\lambda n (A)=0 \text{ for any } n\in \n \} 
\end{equation}
 of $\mathcal{V}(\g,f,k)$ is well-defined.  Moreover, it is a Poisson vertex algebra endowed with the $\lambda$-bracket induced from that of $S(Cur_k(\g))$. (See Proposition \ref{Prop:Welldefined_affine}.)
 
\begin{defn} \label{Def:Walg_2}
 The {\it classical affine W-algebra} $\WW_2(\g,f,k)$  is the PVA defined in (\ref{Eqn:classicalW}) endowed with the  product  and the $\lambda$-bracket 
 \[ (A+I) \cdot (B+I)= (AB)+I , \quad \{A+I \, _\lambda B+I\} = \{A_\lambda B\}+I[\lambda] , \quad A+I,B+I\in \WW(\g,f,k).\]
\end{defn}

\begin{prop} \label{Prop:Welldefined_affine}
The product and the $\lambda$-bracket in Definition \ref{Def:Walg_2} are well-defined.
\end{prop} 
\begin{proof}
To see the well-definedness of the PVA $\WW_2(\g,f,k)$, we have to check that  the algebra is closed under the product and the $\lambda$-bracket. We can check this as follows:\\
(1) By the Leibniz rule, the element $AB+I$ is in $\WW_2�(\g,f,k)$ if $A+I$ and  $B+I$ are in $\WW_2(\g,f,k).$\\
(2) By the definition of a $\WW$-algebra, we have $\{A_\lambda I\} =\{I_\lambda A\}=0+ I[\lambda]$. Moreover, by the Jacobi identity, the element  $\{A_\lambda B\}+I[\lambda]$ is in $\WW_2(\g,f,k)[\lambda]$ if $A+I$ and $B+I$ are in $\WW_2(\g,f,k).$
\end{proof}

\subsection{Equivalence of the definitions of an affine classical W-algebra} \ \\
 
Recall the LCA  $R=Cur_k(\g)\oplus R_{ch}\oplus R_{ne}$. Let us consider the building block 
\[ J_a = a+ \sum_{\alpha \in S} \phi^\alpha \phi_{[u_\alpha, a]}\in S(R), \quad a\in \g.\]
Then 
\begin{equation} \label{Eqn:3.14_0907}
\begin{aligned}
& d_{(0)}(J_a)  = \sum_{\alpha\in S} s(u_\alpha)\phi^\alpha [u_\alpha, a] +  k\partial \sum_{\alpha \in S} s(u_\alpha)(u_\alpha| a) \phi^\alpha \\
& - \sum_{\alpha \in S} s(u_\alpha) \phi^\alpha \left( \pi_+ [u_\alpha, a] + ([u_\alpha, a] | f) + s([u_\alpha, a]) \Phi_{[u_\alpha, a]} \right) \\
& -\sum_{\alpha,\beta\in S} s(u_\alpha) \phi^\alpha \phi^\beta \phi_{[u_\beta, \pi_+[u_\alpha, a]]} +\frac{1}{2} \sum_{\alpha, \beta \in S} s(u_\beta) \phi^\beta \phi^{[u_\beta, u^\alpha]}\phi_{[u_\alpha, a]}.
\end{aligned}
\end{equation}
Here we recall that $s(a)=(-1)^{p(a)}$ for a homogeneous element $a.$

Since $\sum_{\gamma\in S} (u_\gamma |[u_\beta, u^\alpha]) u^\gamma  = \pi_-[u_\beta, u^\alpha]$, we have 
\begin{equation}
\sum_{\alpha, \beta\in S} s(u_\beta)\phi^\beta \phi^{[u_\beta, u^\alpha]} \phi_{[u_\alpha, a]}= \sum_{\beta, \gamma \in S} s(u_\beta) \phi^\beta \phi^\gamma \phi_{[[u_\gamma, u_\beta], a]}
\end{equation}
and
\begin{equation} \label{Eqn:4.12_150213}
[[u_\gamma, u_\beta], a]= [u_\gamma, [u_\beta, a]]- (-1)^{p(\beta)p(\gamma)}[u_\beta, [u_\gamma, a]]. 
\end{equation}
Hence
\begin{equation} \label{Eqn:3.17_0907}
\begin{aligned}
 \sum_{\alpha, \beta \in S} s(u_\beta) \phi^\beta \phi^{[u_\beta, u^\alpha]}\phi_{[u_\alpha, a]}
&= \sum_{\beta, \gamma \in S} s(u_\beta) \phi^\beta \phi^\gamma \phi_{[[u_\gamma, u_\beta], a]} \\
 &= 2 \sum_{\beta, \gamma \in S} s(u_\beta) \phi^\beta \phi^\gamma \phi_{[u_\gamma, [u_\beta, a]]}.
\end{aligned}
\end{equation}
By (\ref{Eqn:3.14_0907}) and (\ref{Eqn:3.17_0907}), we have
\begin{equation}
d_{(0)}(J_a)= \sum_{\alpha\in S}s(u_\alpha) \phi^\alpha K_{[u_\alpha, a]} + \sum_{\alpha \in S} k s(u_\alpha)  (u_\alpha|a) \partial\phi^\alpha
\end{equation}
where, for the projection map $\pi_{\leq}: \g \to \bigoplus_{i\leq 0} \g(i),$
\begin{equation} \label{Eqn:K}
K_a= J_{\pi_\leq a}  - s(a) \Phi_a -(a|f), \qquad a \in\g.
\end{equation}

Also, we have

\begin{equation}
\begin{aligned}
& \{J_a \, _\lambda  J_b\}  = \sum_{\alpha, \beta \in S} \{ a+ \phi^\alpha \phi_{[u_\alpha, a]} \, _\lambda b+ \phi^\beta \phi_{[u_\beta, b]} \} \\
&=  \{a _\lambda b\} + \sum_{\alpha, \beta \in S } \left[ \phi^\alpha \{ \phi_{[u_\alpha, a]} \, _\lambda  \phi^\beta \}\phi_{[u_\beta, b]}  - (-1)^{p(a)p(b)}  \phi^\beta \{ \phi_{[u_\beta, b]} \, _\lambda \phi^\alpha \} \phi_{[u_\alpha, a]}\right] \\
& = \{ a_\lambda b\} +\sum_{\alpha \in S} \phi^\alpha \phi_{[\pi_+ [u_\alpha, a],b]}-(-1)^{p(a)p(b)} \sum_{\beta\in S}\phi^\beta\phi_{[\pi_+[u_\beta, b],a]}\\
& = J_{[a,b]} + \lambda k (a|b) - \sum_{\alpha\in S} \left[\phi^\alpha \phi_{[\pi_\leq [u_\alpha, a],b]}-(-1)^{p(a)p(b)} \sum_{\alpha\in S}\phi^\alpha \phi_{[\pi_\leq[u_\beta, b],a]}\right].
\end{aligned}
\end{equation}
Hence if $a$ and $b$ are both in $\bigoplus_{i\geq 0} \g(i)$ or both in $\bigoplus_{i\leq 0} \g(i)$, then $\{J_a \, _\lambda J_b\} = J_{[a,b]}+k\lambda(a|b).$ Also, we can easily check that
\begin{equation} \label{Eqn:K_brac}
\{K_a\, _\lambda K_b\} = \left\{
\begin{array}{ll}
 K_{[a,b]} + \lambda k(a|b)  \qquad & \text{ if } a, b\in \bigoplus_{i\leq 0}\g(i) , \\
 - K_{[a,b]} =([a,b]|f)  \qquad & \text{ if } a,b \in \g\left(\frac{1}{2}\right),\\
 0 & \text{ otherwise}.
 \end{array} \right.
\end{equation}

Let us denote
\begin{equation}
r_+= \phi_\n \oplus d_{(0)} (\phi_\n), \qquad r_-= J_{\g_\leq} \oplus \phi^{\n_-} \oplus \Phi_{\g(1/2)}
\end{equation} 
and 
\begin{equation}
R_+=\CC[\partial]\otimes r_+, \qquad R_-= \CC[\partial]\otimes r_-. 
\end{equation}
Then $d_{(0)}|_{S(R_+)} \subset S(R_+)$ and $d_{(0)}|_{S(R_-)} \subset S(R_-).$

\begin{prop} \label{Prop:3.8_0914}
Let $d=d_{(0)}|_{S(R_-)}$. Then we have
\begin{equation}
H(S(R),d_{(0)})= H(S(R_-), d).
\end{equation}
Hence $\WW_1(\g,f,k)= H(S(R_-), d).$
\end{prop}
\begin{proof}
Let us denote $d_+=d_{(0)}|_{S(R_+)}$. Then $S(R)= S(R_+) \otimes S(R_-)$ and $d_{(0)} =d_+ \otimes 1 + 1\otimes d.$ Hence, by  K\"unneth lemma,  $H(S(R), d_{(0)})= H(S(R_+), d_+) \otimes H(S(R_-), d).$ Also, by K\"unneth lemma, we have $H(S(R_+), d_+) = S(H(R_+, d_+))= \CC$. Hence $H(S(R), d_{(0)})=H(S(R_-), d).$ 
\end{proof}

Let us define the $\frac{1}{2} \ZZ$-bigrading  on $S(R_-)$:
\begin{equation} \label{Eqn:bigrading}
\text{gr}(J_\alpha)= (j_\alpha-1/2, -j_\alpha+1/2), \quad \text{gr}(\phi^\beta)=(-j_\beta+1/2, j_\beta+1/2)  
\end{equation}
  and $\text{gr}(\Phi_\gamma)=\text{gr}(\partial)=(0,0),$ where $J_\alpha=J_{u_\alpha}$, $\phi^\beta= \phi^{u^\beta}$, $\Phi_\gamma= \Phi_{u_\gamma}$ and $u_\alpha \in \g(j_\alpha)$, $u_\beta \in \g(j_\beta)$, $u_\gamma \in \g(j_\gamma).$ For the first component of the bigrading (\ref{Eqn:bigrading}), we call by $p$-grading and for the second component, we call by $q$-grading. The charge on $S(R_-)$ is defined by the sum of $p$-grading and $q$-grading. Hence
\begin{equation} 
\text{charge}(J_\alpha)=0, \quad \text{charge}(\phi^\beta)=1, \quad  \text{charge}(\Phi_\gamma)=0.
\end{equation}
Consider the decreasing filtration with respect to the $p$-grading 
\begin{equation} \label{Eqn:filt}
 \cdots \subset F_{p+\frac{1}{2}} \subset F_{p} \subset F_{p-\frac{1}{2}} \subset \cdots .
 \end{equation}

Using the facts in Section \ref{Subsec:FC}, we obtain the graded differential $d^{gr}:S(R_-)\to S(R_-)$ such that  
\[ d^{gr}(J_a)=-\sum_{\alpha\in S}s(a)\phi^\alpha(u_\alpha|[a,f]), \quad d^{gr}(\Phi_{[a]})=\sum_{\alpha\in S(1/2)} \phi^\alpha(u_\alpha|[a,f]), \quad d^{gr}(\phi^a)=0.\]
Let us denote $ J_{\g_f}:=\{ \, J_a\, |\, [a,f]=0\,\}.$ Then $\ker(d^{gr}|_{R_-})=\CC[\partial]\otimes J_{\g_f} \oplus \CC[\partial]\otimes \phi^{\n_-}$ and $\text{im}(d^{gr}|_{R_-})=\CC[\partial]\otimes \phi^{\n_-}$. Hence we have 
\[H(R_-, d^{gr})=\CC[\partial]\otimes J_{\g_f}.\]

\begin{lem} \label{Lem:4.8_150307}We have the following properties: 
\begin{enumerate}
\item The complex $(S(R_-), d)$ is a direct sum of locally finite complexes.
\item $H^{pq}(S(R_-), d^{\text{gr}})= 0$ if $p+q\neq0$.
\end{enumerate}
\end{lem}

\begin{proof}
(1) Recall that we have the Hamiltonian operator $H$ on $S(R)$ defined in (\ref{Eqn:Ham}). Since 
\[ \Delta_{J_\alpha}= 1-j_\alpha \text{ for } J_\alpha = J_{u_\alpha},\]
the operator $H|_{S(R_-)}$ is a Hamiltonian operator on $S(R_-).$
Since $d$ preserves the conformal weight and each eigenspace
\[ S(R_-)(i)= \{ a \in S(R_-) | \Delta_a=i\}\subset S(R_-)\]
is finite dimensional, we conclude that the complex $(S(R_-),d)=\bigoplus_{i\in \frac{\ZZ}{2}}(S(R_-)(i),d)$ is a direct sum of locally finite complexes. \\
(2) By K\"unneth lemma, we have
\[ H(S(R_-), d^{gr})= S(H(R_-,d^{gr}))= S(\CC[\partial]\otimes J_{\g_f}).\]
Since any element in $ S(\CC[\partial]\otimes J_{\g_f})$ has charge $0$, we proved the lemma. 
\end{proof}

\begin{prop} \label{Prop:main}
\begin{enumerate}
\item
$\text{gr}^{pq} H(S(R_-), d) \simeq H^{pq}(S(R_-), d^{\text{gr}}).$
\item
$H(S(R_-), d)=H^0(S(R_-),d) \simeq S(\CC[\partial]\otimes J_{\g_f})$ as associative superalgebras.
\end{enumerate}
\end{prop}

\begin{proof}
By Lemma \ref{Lem:4.8_150307}, we have $\text{gr}^{pq} H(S(R_-)(i), d)\simeq H^{pq}(S(R_-)(i), d^{\text{gr}}).$ By taking direct sum $\bigoplus_{i\in \frac{\ZZ}{2}}$ to the both sides, we get $\text{gr}^{pq} H(S(R_-), d) \simeq H^{pq}(S(R_-), d^{\text{gr}}).$ Also, by  Lemma \ref{Lem:4.8_150307} (2), we obtain the second assertion.
\end{proof}

\begin{thm} \label{Thm:main}
Consider the associative superalgebra homomorphism 
\begin{equation}
\bar{f}: S(R_-) \to \mathcal{V}(\g,f,k),  
\end{equation}
such that $K_a \mapsto a$ for $a \in \bigoplus_{i\leq 1} \g(i)$ and $\phi^{n_-} \mapsto 0$ for $n_- \in \n_-$.
Then we have
\begin{enumerate}
\item The map 
\begin{equation}
f: \WW_1(\g,f,k)=H(S(R_-), d) \to \WW_2(\g,f,k)
\end{equation}
is a well-defined superalgebra isomorphism.
\item Moreover, $f$ is a PVA isomorphism.
\end{enumerate}
\end{thm}

\begin{proof}
(1) Since $H(S(R_-), d)= H^0(S(R_-), d)$, any element in $H(S(R_-),d)$ has a representative in $S(\partial^n K_a)$ for $n\in \ZZ_{\geq 0}$ and $a \in \bigoplus_{i\leq \frac{1}{2}} \g(i).$

Now, let us prove that the map $f$ is a well-defined isomorphism. In order to do that, we observe that 
\begin{equation} \label{Eqn:def1_1}
d(K_a) =\sum_{\alpha \in S} \psi^\alpha K_{[u_\alpha, a]}+\sum_{\alpha \in S} k \partial (u_\alpha|a) \psi^\alpha
\end{equation}
where $\psi^\alpha = s(u_\alpha)\phi^\alpha$ and
\begin{equation} \label{Eqn:def1_2}
\begin{aligned}
& d(\partial^{n_a}K_a \cdot \partial^{n_b}K_b)  \\
& = s(a) \partial^{n_a} K_a \cdot \sum_{\alpha \in S} \partial^{n_b} (\psi^\alpha K_{[u_\alpha, b]}+ k \partial (u_\alpha|b) \psi^\alpha) + \sum_{\alpha \in S}\partial^{n_a} (\psi^\alpha K_{[u_\alpha, a]}+ k \partial (u_\alpha|a) \psi^\alpha) \cdot \partial^{n_b}K_b \\
& = \sum_{\alpha\in S} (-1)^{p(a)p(\alpha)}\left[\sum_{i=0}^{n_b}\left( { n_b \choose i} \partial^i \psi^\alpha \cdot\partial^{n_a}K_a\cdot \partial^{n_b-i}K_{[u_\alpha, b]}\right) + k  (u_\alpha|b) \partial^{n_b+1}\psi^\alpha \cdot \partial^{n_a}K_a \right] \\
& +\sum_{\alpha \in S} \left[ \sum_{j=0}^{n_a} \left( {n_a \choose j} \partial^j \psi^\alpha \cdot \partial^{n_a-j} K_{[u_\alpha, a]}\cdot \partial^{n_b} K_b\right)+ k (u_\alpha|a)\partial^{n_a+1} \psi^\alpha \cdot \partial^{n_b}K_b\right].
\end{aligned}
\end{equation} 
On the other hand, we have
\begin{equation}\label{Eqn:def2_1}
 \{u_\alpha\, _\lambda a\}= [u_\alpha, a]+k \lambda (u_\alpha|a)
\end{equation}
and
\begin{equation}\label{Eqn:def2_2} 
\begin{aligned}
&  \{u_\alpha\, _\lambda\, \partial^{n_a}a\cdot \partial^{n_b}b\} \\
 &= (-1)^{p(a)p(\alpha)} \partial^{n_a} a \cdot (\lambda+\partial)^{n_b}\{u_\alpha \, _\lambda \, b\}+ (\lambda+\partial)^{n_a} \{u_\alpha \, _\lambda \, a\} \cdot \partial^{n_b}b\\
 & = (-1)^{p(a)p(\alpha)}  \sum_{i=0}^{n_b}  \partial^{n_a} a \cdot \left[{n_b \choose i}\lambda^i \partial^{n_b-i} [u_\alpha,  b]  +  k \lambda^{n_b+1}(u_\alpha|b) \right]\\
 & +  \sum_{j=0}^{n_a} \left[ {n_a \choose j}  \lambda^j \partial^{n_a-j} [u_\alpha, a] +  k  \lambda^{n_a+1} (u_\alpha|a) \right] \cdot \partial^{n_b}b.
\end{aligned}
\end{equation}

Let us  denote 
\begin{equation}
 K_{AB}=K_A K_B, \quad K_{\partial A}= \partial K_A, \quad K_{A+B}=K_A+K_B, \quad K_C=C  
\end{equation}
for  $A,B \in S ( \CC[\partial]\otimes \bigoplus_{i\leq 1} \g(i) )$ and $C\in \CC$. Assume that 
\[  \{ u_\alpha\, _\lambda A\} = \sum_{i\geq 0} \frac{\lambda^i}{i!} (u_{\alpha (i)} A), \quad  \{ u_\alpha\, _\lambda B\} = \sum_{i\geq 0} \frac{\lambda^i}{i!} (u_{\alpha (i)} B) \] 
 for some $ u_{\alpha (i)} A$, $u_{\alpha(i)}B \in S ( \CC[\partial]\otimes \g)$
 and
\[ d(K_A)= \sum_{i\geq 0} \sum_{\alpha \in S} \frac{\partial^i  \psi^\alpha}{i!}  K_{u_{\alpha (i)} A}, \quad d(K_B)= \sum_{i\geq 0}  \sum_{\alpha \in S} \frac{\partial^i  \psi^\alpha}{i!}  K_{u_{\alpha (i)} B}.\]
Then 
\begin{equation}
\begin{aligned}
& \{u_\alpha \, _\lambda AB\}  =\sum_{i\geq 0} \left(\frac{(\lambda+\partial)^i}{i!} u_{\alpha(i)}A\right) B +(-1)^{p(\alpha)p(A)} A \sum_{i\geq 0} \left(\frac{(\lambda+\partial)^i}{i!} u_{\alpha(i)}B\right) \\
& =\sum_{i\geq 0}\sum_{j=0}^i \frac{1}{i!} {i \choose j} \lambda^j \left[ (\partial^{i-j} u_{\alpha(i)} A)B 
+ (-1)^{p(\alpha)p(A)} A (\partial^{i-j} u_{\alpha(i)} B)\right]
\end{aligned}
\end{equation}
and 
\begin{equation}
\begin{aligned}
& d(K_A K_B) = (-1)^{p(A)} K_A d(K_B) + d(K_A) K_B \\
&=  \sum_{i\geq 0} \sum_{\alpha \in S} (-1)^{p(A)} K_A \frac{\partial^i}{i!} (\psi^\alpha   K_{u_{\alpha(i)}B}) 
+\sum_{i\geq 0} \sum_{\alpha \in S}  \frac{\partial^i}{i!} (\psi^\alpha   K_{u_{\alpha(i)}A}) K_B\\
& = \sum_{i\geq 0}\sum_{\alpha \in S} \sum_{j=0}^i \frac{1}{i!} {i \choose j}  \partial^j \psi^\alpha  \left[ (-1)^{p(\alpha)p(A)}   K_A  (\partial^{i-j}   K_{u_{\alpha(i)}B}) +     (\partial^{i-j}   K_{u_{\alpha(i)}A}) K_B\right].
\end{aligned}
\end{equation}
Hence 
\[\{ u_\alpha\, _\lambda AB\} = \sum_{i\geq 0} \frac{\lambda^i}{i!} (u_{\alpha (i)} AB) \Longleftrightarrow d(K_{AB})= \sum_{i\geq 0} \sum_{\alpha \in S} \frac{\partial^i  \psi^\alpha }{i!} K_{u_{\alpha (i)} AB}.\]
Since we have (\ref{Eqn:def1_1}), (\ref{Eqn:def1_2}), (\ref{Eqn:def2_1}), and (\ref{Eqn:def2_2}), we conclude that 
\[  \{ u_\alpha\, _\lambda A\} = \sum_{i\geq 0} \frac{\lambda^i}{i!} (u_{\alpha (i)} A)\quad \Longleftrightarrow \quad  d(K_A)= \sum_{i\geq 0} \sum_{\alpha \in S} \frac{\partial^i \psi^\alpha }{i!}  K_{u_{\alpha (i)} A}\]
for any $A \in S(\CC[\partial]\otimes(\bigoplus_{i\leq 1} \g(i))).$
Therefore the followings are equivalent
\begin{enumerate}[(i)]
\item $A\in \WW_2(\g,f,k)$,
\item $u_{\alpha(i)}A=0$ in $S(\CC[\partial]\otimes \g)/I$ for any $\alpha \in S$ and $i\geq 0$, 
\item $K_{u_{\alpha(i)} A}=0$ for any $i\geq 0$ and $\alpha \in S$,
\item $K_A\in \WW_1(\g,f,k).$
\end{enumerate} \ \\
(2) 
Let $\sum_{j\in J} A_j M_j$ and $\sum_{k\in K} B_k N_k$ be elements in $\WW_2(\g,f,k)$, such that $A_j, B_k \in S(\CC[\partial]\otimes (\bigoplus_{i\leq 0} \g(i)))$ and $M_j, N_k \in S(\CC[\partial]\otimes \g(\frac{1}{2})).$ Then
\begin{equation} \label{Eqn:brack_1}
\begin{aligned}
& \{\sum_{j\in J} A_j M_j \, _\lambda \sum_{k\in K} B_k N_k \} \\
& = \sum_{j\in J, k\in K} (-1)^{p(N_k)(p(A_j)+p(M_j)+p(B_k))+p(M_j)p(B_k)} N_k\{A_j\, _{\lambda+\partial}B_k\}_{\to}M_j \\
&\textstyle -\sum_{j\in J, k\in K}(-1)^{p(M_j)p(B_k)+p(A_j)(p(N_k)+p(M_j)+p(B_k))}B_k\{M_j\, _{\lambda+\partial}N_k\}_{\to}A_j.
\end{aligned}
\end{equation}
Here, we used that $\sum_{j\in J} \{A_jM_j\, _\lambda N_k\}= \sum_{k\in K} \{ M_j\, _\lambda B_k N_k\}=0.$\\
On the other hand, 
\begin{equation} \label{Eqn:brack_2}
\begin{aligned}
& \{\sum_{j\in J} K_{A_j} K_{M_j} \, _\lambda \sum_{k\in K} K_{B_k} K_{N_k} \} \\
& = \sum_{j\in J, k\in K} (-1)^{p(N_k)(p(A_j)+p(M_j)+p(B_k))+p(M_j)p(B_k)}K_{N_k}\{K_{A_j} \, _\lambda  K_{B_k}\} K_{M_j}\\
&\textstyle+ \sum_{j\in J, k\in K}(-1)^{p(M_j)p(B_k)+p(A_j)(p(N_k)+p(M_j)+p(B_k))}K_{B_k}\{K_{M_j}\, _{\lambda+\partial}K_{N_k}\}_{\to}K_{A_j}.
\end{aligned}
\end{equation}

We can see that there exist $c_i \in \mathcal{V}(\g,f,k)$ for $i\geq 0$ such that 
\[\{A_j \,_\lambda B_k\} =\sum_{i\geq 0} c_i \lambda^i \in \mathcal{V}(\g,f,k)[\lambda]\quad  \Longleftrightarrow \quad \{K_{A_j}\, _\lambda K_{B_k}\} = \sum_{i\geq 0} K_{c_i} \lambda^i\]
and there exist $C_i\in \mathcal{V}(\g,f,k)$ for $i\geq 0$ such that 
\[\{M_j \,_\lambda N_k\} =\sum_{i\geq 0} C_i \lambda^i \in \mathcal{V}(\g,f,k)[\lambda]\quad  \Longleftrightarrow \quad \{K_{M_j}\, _\lambda K_{N_k}\} = -\sum_{i\geq 0} K_{C_i} \lambda^i.\]
The last equivalence comes from the fact that $\{n_1\, _\lambda n_2\} = -(f|[n_1, n_2])$ and $\{K_{n_1}\, _\lambda K_{n_2}\} = (f|[n_1, n_2])$ for $n_1, n_2 \in \g(\frac{1}{2}).$ Hence (\ref{Eqn:brack_1}) and (\ref{Eqn:brack_2}) imply that the map $f$ is a PVA isomorphism.

\end{proof}

\begin{prop} \label{Prop:generator_aff}
Suppose $\{v_\alpha\}_{\alpha\in J}$ is a basis of $\g_f$ such that $v_\alpha \in \g(j_\alpha)$ and $v_\alpha$ has the conformal weight $\Delta_\alpha$. If we have a subset $A=\{v_\alpha+a_\alpha |\alpha\in J\}\subset \WW(\g,f,k)$ such that  
\begin{equation} \label{a_assump}
a_\alpha\in S(\CC[\partial]\otimes(\bigoplus_{i>j_\alpha} \g(i)))
\end{equation}
 then $A$ is a set of  free generators of $\WW(\g,f,k).$
\end{prop}

\begin{proof}
We already showed in Proposition \ref{Prop:main} that there is a set of free generators $\{v_\alpha+ b_\alpha|\alpha \in J\}$ of $\WW(\g,f,k)$ such that $\gr(v_\alpha+b_\alpha)=v_\alpha.$ In other words, using the filtration (\ref{Eqn:filt}) and Theorem \ref{Thm:main}, we have   $v_\alpha \in F_{j_\alpha-\frac{1}{2}} \backslash F_{j_\alpha}$ and $b_\alpha \in F_{j_\alpha}.$ It is not hard to see that $F_{j_\alpha} \subset  S(\CC[\partial]\otimes(\bigoplus_{i>j_\alpha} \g(i))).$ Hence we proved that the existence of such generating sets.

Now let us assume there is another subset  $A=\{v_\alpha+a_\alpha |\alpha\in J\}\subset \WW(\g,f,k)$ satisfying (\ref{a_assump}). We denote by $B_\alpha = v_\alpha+ b_\alpha$ and $A_\alpha=v_\alpha+a_\alpha$. 
If $v_\alpha \in \g(0)$ then $a_\alpha-b_\alpha\in S(\CC[\partial]\otimes(\bigoplus_{i>0} \g(i))).$ However, since $a_\alpha-b_\alpha\in \WW(\g,f,k)$ and $\WW(\g, f,k)\cap S(\CC[\partial]\otimes(\bigoplus_{i>0} \g(i)))=0$, we have $A_\alpha=B_\alpha$ and 
\[\CC_{\text{diff}}[A_\alpha|\, \alpha\in J, \,  v_\alpha \in \g_f(0)]= \CC_{\text{diff}}[B_\alpha\, |\, \alpha\in J, \,  v_\alpha \in \g_f(0)].\]
 Here $\CC_{\text{diff}}[S]$ for a set $S$ denotes supersymmetric algebra generated by $\{\partial^n s\, |\, n\in \ZZ_+,\, s\in S\}.$

Suppose $\CC_{\text{diff}}[A_\alpha|\, \alpha\in J, \,  v_\alpha \in \bigoplus_{i\geq k} \g_f(k)]= \CC_{\text{diff}}[B_\alpha\, |\, \alpha\in J, \,  v_\alpha \in \bigoplus_{i\geq k} \g_f(k)]$ for some $k\leq 0.$ Let $v_\beta\in \g_f(k-\frac{1}{2})$ for some $\beta\in J$. Since $A_\beta-B_\beta \in \WW(\g,f,k)\cap S(\CC[\partial]\otimes(\bigoplus_{i\geq k} \g(i)))$ and 
\begin{equation*}
\begin{aligned}
\textstyle \WW(\g,f,k)\cap S(\CC[\partial]\otimes(\bigoplus_{i\geq k} \g(i)))& \textstyle   = \CC_{\text{diff}}[A_\alpha|\, \alpha\in J, \,  v_\alpha \in \bigoplus_{i\geq k} \g_f(k)] \\
  & \textstyle   = \CC_{\text{diff}}[B_\alpha\, |\, \alpha\in J, \,  v_\alpha \in \bigoplus_{i\geq k} \g_f(k)]
 \end{aligned}
 \end{equation*}
 we have 
 \begin{equation*}
 \begin{aligned}
&  \textstyle A_\beta \in  \CC_{\text{diff}}[B_\beta, B_\alpha\, |\, \alpha\in J, \,  v_\alpha \in \bigoplus_{i\geq k} \g_f(k)], \\
 &\textstyle B_\beta \in  \CC_{\text{diff}}[A_\beta, A_\alpha\, |\, \alpha\in J, \,  v_\alpha \in \bigoplus_{i\geq k} \g_f(k)].
 \end{aligned}
 \end{equation*}
 Hence 
 \[ \textstyle \CC_{\text{diff}}[ B_\alpha\, |\, \alpha\in J, \,  v_\alpha \in \bigoplus_{i\geq k-\frac{1}{2}} \g_f(k)]= \CC_{\text{diff}}[A_\alpha\, |\, \alpha\in J, \,  v_\alpha \in \bigoplus_{i\geq k-\frac{1}{2}} \g_f(k)].\]
 By an induction, we have $\CC_{\text{diff}}[v_\alpha+a_\alpha|\alpha\in J]=\CC_{\text{diff}}[v_\alpha+b_\alpha|\alpha\in J]=\WW(\g,f,k).$ Hence $A$ is a set of generators of $\WW(\g,f,k).$
\end{proof}

\vskip 3mm

\section{Relations between finite and affine W-superalgebras} \label{Sec:finite}

\subsection{Finite W-superalgebras}\ \\

For finite W-superalgebras, we can obtain an analogous result to Theorem \ref{Thm:main} and Proposition \ref{Prop:generator_aff}. (cf. appendix of \cite{DK} and \cite{GG})

\begin{defn} \label{Def:QF_W_1}
Let  $(\phi^{\n_-} \oplus \phi_\n)$ and $\Phi_{\n/\m}$ be nonlinear Lie superalgebras such that 
\begin{enumerate}
\item as vector superspaces
\[ \phi^{\n_-} \simeq \Pi(\n_-), \quad \phi_{\n} \simeq \Pi(\n), \quad \Phi_{[\n, \m]} \simeq \g(1/2)\]
 where $\Pi$ denotes parity reversing, 
\item for $a,b \in \n_-$ and $c,d, n_1, n_2 \in \n,$ 
\[ [\phi^a, \phi^b]=[\phi_c,\phi_d]=0, \quad [\phi_c, \phi^a]=(c|a), \quad [\Phi_{n_1}, \Phi_{n_2}]= (f| [n_1, n_2]).\]
\end{enumerate}
Let us denote  $r=\g\oplus (\phi^{\n_-} \oplus \phi_\n) \oplus \Phi_{\n/\m}$ and let 
\[d= \sum_{\alpha \in S} s(\alpha) \phi^\alpha u_\alpha + \sum_{\alpha \in S(1/2)} \phi^\alpha \Phi_\alpha + \phi^f +\frac{1}{2} \sum_{\alpha,\beta \in S} \phi^\alpha \phi^\beta \phi_{[u_\beta, u_\alpha]} \in U(r) \]
for $\phi^\alpha = \phi^{u^\alpha}$, $s(\alpha)=s(u_\alpha)$ and $\Phi_\alpha = \Phi_{u_\alpha}.$ If the adjoint map $\ad\, d :U(r) \to U(r)$ is defined by the Lie bracket on $r$ and Leibniz rules, the associative superalgebra
\[ W^{fin}_1(\g,f)= H(U(r), \ad\, d)\]
is called the quantum finite W-superalgebra associated to $\g$ and $f$.
\end{defn}

In order to see Definition \ref{Def:QF_W_1} makes sense, we have to show the following lemma.

\begin{lem}
\begin{enumerate}
\item We have $(\ad\, d)^2=0$ and $d$ is an odd element in $U(r).$ 
\item The associative product of $U(r)$ induces the product of $W^{fin}_1.$
\end{enumerate}
\end{lem}
\begin{proof}
The proof is almost same as that of the lemma in the affine classical W-superalgebra case.
\end{proof}

We introduce another definition of quantum finite W-superalgebras. Let $U(\g)$ be the universal enveloping algebra of $\g$ and consider the Lie bracket $[\, , \, ]$ on $U(\g)$ defined by the Lie bracket on $\g$ and Leibniz rules. 

\begin{defn} \label{Def:QF_W_2}
Let $I^{fin}$ be the associative superalgebra ideal of $U(\g)$ generated by $\{m +(f|m) | m\in \m\}.$ We denote 
\[ \mathcal{U}(\g,f)=U(\g)/I^{fin}.\]
The adjoint action of $\n$ on $\mathcal{U}(\g,f)$ is defined by $\ad\, n\, (A)= [a,A]$ and the invariant space $(U(\g)/I^{fin})^{\ad\, \n}$ is called the {\it quantum finite W-superalgebra} associated to $\g$ and $f$. Also, we write 
\[ W^{fin}_2(\g,f)= \mathcal{U}(\g,f)^{\ad\, \n}\]
and the associative product of $W^{fin}_2$ is defined by 
\[ (A+I^{fin})\cdot (B+I^{fin}) = AB +I^{fin} .\]
\end{defn}

In order to see Definition \ref{Def:QF_W_2} makes sense, we have to show the following lemma.

\begin{lem}
\begin{enumerate}
\item We have $[n, I^{fin}] \in I^{fin}$.
\item If $A$ and $B$ are in $W^{fin}_2(\g,f)$ then $AB$ is  in $W^{fin}_2(\g,f).$
\end{enumerate}
\end{lem}

\begin{proof}
The proof is almost same as that of the lemma in the affine classical W-superalgebra case.
\end{proof}

Let us consider the building blocks
\[ J_a= a+\sum_{\alpha\in S}\phi^\alpha\phi_{[u_\alpha, a]}\qquad \text{ for } a\in \g,\]
where $\phi_g=\phi_{\pi_+  g}$  and $\pi_+:\g\to \n$ is the canonical projection map.  We denote  
\begin{equation}
r_+= \phi_\n \oplus \ad \, d\,  (\phi_\n)\quad \text{ and }\quad   r_-= J_{\g_\leq} \oplus \phi^{\n_-} \oplus \Phi_{\g(1/2)}.
\end{equation} 
Then $\ad \, d\, |_{U(r_+)} \subset U(r_+)$ and $\ad\, d\, |_{U(r_-)} \subset U(r_-).$  As a consequence, we get the following proposition.

\begin{prop}
Let $d=\ad\, d\, |_{U(r_-)}$. Then we have
\begin{equation}
H(U(r),\ad \,  d)= H(U(r_-), d).
\end{equation}
Hence $\WW_1(\g,f,k)= H(U(r_-), d).$
\end{prop}

\begin{proof}
The proof is almost same as that of the Proposition \ref{Prop:3.8_0914} in the affine classical W-superalgebra case.
\end{proof}

For
\begin{equation} \label{Eqn:K}
K_a= J_{\pi_\leq a}  - s(a) \Phi_a -(a|f), \qquad a \in \bigoplus_{i\leq 1} \g(i),
\end{equation}
we have the following theorem.

\begin{thm} \label{Thm:main_2}
Consider the associative superalgebra homomorphism 
\begin{equation}
\bar{f}: U(r_-) \to \mathcal{U}(\g,f),  
\end{equation}
such that $K_a \mapsto a$ for $a \in \bigoplus_{i\leq 1} \g(i)$ and $\phi^{n_-} \mapsto 0$ for $n_- \in \n_-$. Then 
\begin{equation}
f: W_1^{fin}(\g,f)=H(U(r_-), d_-) \to W^{fin}_2(\g,f)
\end{equation}
is a well-defined superalgebra isomorphism.
\end{thm}

\begin{proof}
The proof is almost same as that of Theorem \ref{Thm:main} in the affine classical W-superalgebra case.
\end{proof}

Recall that the bigrading (\ref{Eqn:bigrading}) is defined on $S(R_-)$. Suppose we have the bigrading $\gr$ of $r_-$ which is induced from (\ref{Eqn:bigrading}). We call the first component of of $\gr$ by $p$-grading.

\begin{prop} \label{Prop:generator_fin_qt}
Suppose $\{v_\alpha\}_{\alpha\in J}$ is a basis of $\g_f$ such that $v_\alpha \in \g(j_\alpha)$. If we have a subset $A=\{v_\alpha+a_\alpha |\alpha\in J\}\subset W^{fin}(\g,f)$ such that   $p$-grading of $a_i$ is greater than that of $v_\alpha$,
 then $A$ is a set of  free generators of $W^{fin}(\g,f).$
\end{prop}

\begin{proof}
The proof is omitted here. It can be proved by the analogous proof of Proposition \ref{Prop:generator_aff}.
\end{proof}

By substituting universal enveloping algebras of Lie superalgebras with supersymmetric algebras of Lie superalgebras, and the ideal $I^{fin}$ generated by the subset $\{m+(f|m)|m\in \m\}$ of $U(\g)$ by  the ideal $\mathcal{I}^{fin}$ generated by the subset $\{m+(f|m)|m\in \m\}$ of $S(\g)$, we get the following theorem about the classical finite W-algebra $\WW^{fin}(\g,f)$ associated to $\g$ and $f$. Also, we denote 
\[ \mathcal{S}(\g,f)=S(\g)/\mathcal{I}^{fin}.\]

\begin{thm} \label{Thm:main_2}
Consider the associative superalgebra homomorphism 
\begin{equation}
\bar{f}: S(r_-) \to \mathcal{S}(\g,f),  
\end{equation}
such that $K_a \mapsto a$ for $a \in \bigoplus_{i\leq 1} \g(i)$ and $\phi^{n_-} \mapsto 0$ for $n_- \in \n_-$. Then 
\begin{equation}
f: \WW_1^{fin}(\g,f)=H(S(r_-), d_-) \to \WW^{fin}_2(\g,f)=\mathcal{S}(\g,f)^{\ad\, \n}
\end{equation}
is a well-defined Poisson superalgebra isomorphism.
\end{thm}

\begin{prop} \label{Prop:generator_fin}
Take the $\Delta$-grading on $S(\g)$ defined by $\Delta_a=1-j_a$ for $a\in \g(j_a)$.
Suppose $\{v_\alpha\}_{\alpha\in J}$ is a basis of $\g_f$ such that $v_\alpha \in \g(j_\alpha)$. If we have a subset $A=\{v_\alpha+a_\alpha |\alpha\in J\}\subset \WW^{fin}(\g,f)$ such that  
\[a_\alpha \in S(\bigoplus_{i>j_\alpha} \g(i))\]
 then $A$ is a set of  free generators of $\WW^{fin}(\g,f).$
\end{prop}

Also,  classical finite W-superalgebras can be understood as finitizations of classical affine W-superalgebras via classical Zhu map. 


\begin{thm} \label{Thm:aff_fin}
Given Lie superalgebra $\g$ and a nilpotent $f\in \g$, there is a Poisson algebra isomorphism 
\[ Zhu_H(\WW(\g,f,k))\simeq \WW^{fin}(\g,f)\]
  where $H=L_{(1)}$ and $L\in \WW(\g,f,k)$ is the image of  $L_\g=\sum_{\alpha\in \bar{S}}\frac{1}{2k}u^\alpha u_\alpha +\partial x \in S(Cur_k(\g))$ in $\mathcal{V}(\g,f,k).$ 
\end{thm}

\begin{proof}
As we showed in Example \ref{Ex:Zhu_Alge}, the $H$-twisted Zhu algebra $Zhu_H(S(Cur_k(R)))$ of $S(Cur_k(R))=S(\CC[\partial]\otimes \g)$ endowed with the Poisson $\lambda$-bracket 
\[\{a_\lambda b\}=[a,b]+k\lambda(a|b), \qquad a,b\in \g, \quad k\in \CC\backslash\{0\}, \]
is the supersymmetric algebra generated by $\g$ endowed with the Poisson bracket 
  \[ \{v_a,v_b\}=v_{[a,b]}, \qquad v_a=a-k(x|a).\]
  
 If $\phi\in \WW(\g,f,k)$ then $\{n_\lambda \phi\} \in I[\lambda]$ and hence $\{ n, \overline{\phi}\} \in Zhu_H(I)$, for any $n\in\n$ and  $\overline{\phi}\in Zhu_H(\WW(\g,f,k)).$ Hence we have
  \[ \overline{\phi}\in ( Zhu_H(S(Cur_k(\g)))/Zhu_H(I))^{\ad \n}\]

 Observe the following facts: 
 \begin{enumerate}[(i)]
  \item  Suppose $v$ is an associative superalgebra isomorphism $v: S(\g)\to Zhu_H(S(Cur_k(\g)))$ defined by $a\mapsto v_a$ for $a\in \g$. Then $Zhu_H(S(Cur_k(\g)))$ is isomorphic to  $S(\g)$ as Poisson superalgebras via $v$. 
 \item Since $v_m=m$ for $m\in \m$, the ideal $\mathcal{I}^{fin}$ of $S(\g)$ such that $\mathcal{S}(\g,f)=S(\g)/\mathcal{I}^{fin}$ is $Zhu_H(I)$.
 \item Since $n=v_n$ for any $n\in \n,$ the adjoint action $\ad_n$ on $S(\g)$ is same as $\ad_{v_n}$ on $S(\g).$
 \end{enumerate}
 Hence  \[ Zhu_H(S(Cur_k(\g)))/Zhu_H(I))^{\ad \n}\simeq \WW^{fin}(\g,f)\]
and  $Zhu_H(\WW(\g,f,k))\hookrightarrow \WW^{fin}(\g,f).$ On the other hand, we obtain generating sets of $\WW^{fin}(\g,f)$ by taking the image of generators of $\WW(\g,f,k)$ in $Zhu_H(\WW(\g,f,k)).$ (See Proposition \ref{Prop:generator_aff} and Proposion \ref{Prop:generator_fin}.) Hence $\WW^{fin}(\g,f)\simeq Zhu_H(\WW(\g,f,k))$ as Poisson superalgebras via the map $v.$
 \end{proof}

By Example \ref{Ex:Zhu_Alge} and Theorem \ref{Thm:aff_fin}, we have 
\[\xymatrix{
S(Cur_k(\g)) \ar[d]^{Zhu_H} \ar[r] & \mathcal{V}(\g,f,k)  \supset  \WW(\g,f,k)  \ \ \ \ \   \ \ \ar[d]^{Zhu_H} \\
S(\g) \ar[r] & \mathcal{S}(\g,f)   \supset  \WW^{fin}(\g,f)\ \ \ \ \ \  \   }. \]

Moreover, by the following theorem, we can easily obtain $\WW^{fin}(\g,f)$  from $\WW(\g,f,k).$

\begin{thm} \cite{DK} \label{Thm:affine-finite}
Let $(R, \{ \, _\lambda \, \})$ be a  nonlinear Lie conformal algebra and $(S(R), \{ \, _\lambda \, \})$ be the Poisson vertex algebra generated by $R$. Then the $H$-twisted Zhu algebra $Zhu_H(S(R))$ is isomorphic to the supersymmetric algebra $S(R/\partial R)$ endowed with the bracket defined by  $\{ \bar{a} , \bar{b}\}= \overline{\{a_\lambda b\}|_{\lambda=0}},$ where $a,b\in R$ and $\bar{a}, \bar{b}$ are the images of $a,b$ in $R/\partial R.$
\end{thm}

\begin{cor}  \label{Cor:affine-finite}
Let us denote by  $\g_{<1}=\bigoplus_{i<1}\g(i)$ and $G= S(\CC[\partial]\otimes \g_{<1})$. Then the differential algebra $G$ is isomorphic to $\mathcal{V}(\g,f,k)$.  Consider the associative superalgebra homomorphism $p:G\to S(\g_{<1})$ such that 
\[  \partial^n a \mapsto  \delta_{n0}a, \quad \text{ for } a\in g,\, n\in \ZZ_{\geq 0}.\]
Then $p(\WW(\g,f,k))=\WW^{fin}(\g,f).$ The Poisson bracket on $\WW^{fin}(\g,f)$ is defined by 
\[\{p(w_1), p(w_2)\}=p(\{w_{1\, \lambda} w_2\}|_{\lambda=0}).\]
\end{cor}

\begin{proof}
If we take $R=Cur_k(\g)$ in Theorem \ref{Thm:affine-finite} then $Zhu_H(S(R))\simeq S(R/\partial R)$ for the Hamiltnonian operator $H$ in Theorem \ref{Thm:aff_fin}. Denote $M= \{m+(f|m)|m\in \m\}$ then $S(R)/M S(R)\simeq G\simeq \mathcal{V}(\g,f,k)$ as differential algebras and $S(R/\partial R)/MS(R/\partial R)\simeq S(\g_{<1})$ as associative superalgebras. Since $S(R/\partial R)\simeq Zhu_H(S(R))$ and the Poisson bracket on $S(R/\partial R)$ is induced from the $\lambda$-bracket on $S(R)$, it is easy to see that   $p(\WW(\g,f,k))\subset \WW^{fin}(\g,f).$ Moreover, using Proposition \ref{Prop:generator_fin}, we can see that $p(\WW(\g,f,k))= \WW^{fin}(\g,f).$ 

The Poisson $\lambda$ bracket on $\WW(\g,f,k)$ is induced from that of $S(R)$ and Poisson bracket on $\WW^{fin}(\g,f)$ is induced from that of $S(R/\partial R)$. Hence the Poisson bracket on $\WW^{fin}(\g,f)$ is \[\{p(w_1), p(w_2)\}=p(\{w_{1\, \lambda} w_2\}|_{\lambda=0})\] by Theorem  \ref{Thm:affine-finite}.
\end{proof}

\section{Generators of  classical W-superalgebras} \label{Sec:Example}
\subsection{A W-superalgebra associated to a minimal nilpotent $f$}\ \\

Let $f$ be an even minimal nilpotent in $\g$ and let  $\{z_\alpha| \alpha\in S(1/2)\}$ and $\{z^*_\alpha|\alpha\in S(1/2)\}$ be bases of $\g\left(\frac{1}{2}\right)$ such that $[z_\alpha, z^*_\beta]=-e$. Denote by $\ad_\lambda n (A)$ or by $\{ n\, _\lambda\, A \}$ the $\ad_\lambda n$-action on $A\in \mathcal{V}(\g, f,k).$ Then $\{z_\alpha\, _\lambda z^*_\beta\}=-e= \delta_{\alpha,\beta}$ and $\g_f=\g_f(0) \oplus \g(-\frac{1}{2}) \oplus \CC f,$ where $\g_f=\{ g\in \g |[f,g]=0\}$ and $\g_f(0)=\g_f\cap \g(0).$

\begin{prop} \label{Prop:generator}
The affine classical  $\WW$-superalgebra $\WW(\g,f,k)$ has free generators (as a differential algebra) 
\begin{equation}
\begin{aligned}
& \phi_v= v-\frac{1}{2}\sum_{\alpha\in S(1/2)} z^*_\alpha[z_\alpha, v], \\ 
& \phi_w=w-\sum_{\alpha\in S(1/2)} z^*_\alpha[z_\alpha, w]+\frac{1}{3} \sum_{\alpha, \beta\in S(1/2)} z^*_\alpha z^*_\beta[z_\beta, [z_\alpha, w]]-\sum_{\alpha\in S(1/2)} k(z_\alpha| w) \partial z^*_\alpha, \\
& \phi_f= [\text{image of $(-L_\g)$ in $\mathcal{V}(\g,f,k)$}] + k \frac{1}{2}\sum_{\alpha\in S(1/2)}(\partial z^*_\alpha) z_\alpha
\end{aligned}
\end{equation}
where $v\in \g_f(0)$ and $w\in \g(-\frac{1}{2}).$
\end{prop}

\begin{proof}
It is enough to show that $\{z_{\gamma\, \lambda} \phi_v\}= \{z_{\gamma\, \lambda} \phi_w\}= \{z_{\gamma\, \lambda} \phi_f\}=0$ for any $\gamma\in S(1/2)$. Since $\{z_{\gamma \, \lambda} z^*_\alpha\}=\delta_{\alpha,\gamma}$ in $\mathcal{V}(\g,f,k)$, we have 
\begin{equation*}
\begin{aligned}
& \{\ z_{\gamma\ \lambda} \sum_{\alpha\in S(1/2)}z^*_\alpha [z_\alpha, v]\ \}=[z_\gamma,v]+\sum_{\alpha\in S(1/2)}(-1)^{p(\alpha)p(\gamma)} z^*_\alpha[z_\gamma,[z_\alpha, v]]\\
&=[z_\gamma,v]+\sum_{\alpha\in S(1/2)}(-1)^{p(\alpha)p(\gamma)} z^*_\alpha[[z_\gamma,z_\alpha], v]+\sum_{\alpha\in S(1/2)}z^*_\alpha[z_\alpha,[z_\gamma, v]],
\end{aligned}
\end{equation*}
where $p(\alpha)=p(z_\alpha), p(\beta)=p(z_\beta), p(\gamma)=p(z_\gamma)$.
Note $[e,v]=0$. Hence $[[z_\gamma,z_\alpha], v]$=0. Using the fact that $\sum_{\alpha\in S(1/2)}z^*_\alpha[z_\alpha,[z_\gamma, v]]=[z_\gamma,v]$, we conclude 
\begin{equation*}
 \{\ z_{\gamma\ \lambda} \sum_{\alpha\in S(1/2)}z^*_\alpha [z_\alpha, v]\ \}=2[z_\gamma,v]
\end{equation*}
and 
\begin{equation*}
\{z_{\gamma\, \lambda\, }\phi_v\}=\{ \ z_{\gamma\ \lambda} \ v-\frac{1}{2}  \sum_{\alpha\in S(1/2)}z^*_\alpha [z_\alpha, v]\ \}=0.
\end{equation*}
Now we shall show  $\{z_{\gamma\, \lambda} \phi_w\}=0.$ We have
\begin{equation} \label{Eqn:5.1_0708}
\begin{aligned}
& \{z_{\gamma\, \lambda }\, w\}=[z_\gamma, w]+k\lambda(z_\gamma | w),\\
& \{\, z_{\gamma\ \lambda }\, \sum_{\alpha\in S(1/2)} z^*_\alpha[z_\alpha, w] \, \}=[z_\gamma, w]+(-1)^{p(\alpha)p(\gamma)} z^*_\alpha[z_\gamma, [z_\alpha, w]], \\
\end{aligned}
\end{equation}
and
\begin{equation} \label{Eqn:5.2_0708}
\begin{aligned}
& \{ \, z_{\gamma\ \lambda}\, \sum_{\alpha, \beta\in S(1/2)} z^*_\alpha z^*_\beta[ z_\beta, [z_\alpha, w]]\, \} =\sum_{ \beta\in S(1/2)}z^*_\beta[z_\beta, [z_\gamma, w]] \\
 &+\sum_{ \alpha,\beta\in S(1/2)}\left((-1)^{p(\alpha)p(\gamma)} z^*_\alpha[z_\gamma, [z_\alpha, w]] +(-1)^{(p(\alpha)+p(\beta))p(\gamma)} z^*_\alpha z^*_\beta [z_\gamma,[z_\beta, [z_\alpha, w]]]\right).
\end{aligned}
\end{equation}
Since
\begin{equation}
\begin{aligned}
& \sum_{\beta\in S(1/2)}z^*_\beta[z_\beta, [z_\gamma,w]]= \sum_{\beta\in S(1/2)} \left( z^*_\beta[[z_\beta, z_\gamma],w]+ (-1)^{p(\beta)p(\gamma)}z^*_\beta[z_\gamma, [z_\beta,w]]\right); 
\end{aligned}
\end{equation}
and
\begin{equation}
\begin{aligned}
& \sum_{\alpha,\beta\in S(1/2)}(-1)^{(p(\alpha)+p(\beta))p(\gamma)} z^*_\alpha z^*_\beta [z_\gamma,[z_\beta, [z_\alpha, w]]] \\
& = -\sum_{\beta\in S(1/2)} z^*_\beta [[z_\beta, z_\gamma], w] +\sum_{\alpha \in S(1/2)}(-1)^{p(\alpha)p(\gamma)} z^*_\alpha[z_\gamma, [z_\alpha, w]],
\end{aligned}
\end{equation}
the equation (\ref{Eqn:5.2_0708}) implies 
\begin{equation} \label{Eqn:5.5_0708}
 \{ \, z_{\gamma\ \lambda}\, \sum_{\alpha, \beta\in S(1/2)} z^*_\alpha z^*_\beta[ z_\beta, [z_\alpha, w]]\, \} =3\sum_{\alpha \in S(1/2)}(-1)^{p(\alpha)p(\gamma)} z^*_\alpha[z_\gamma, [z_\alpha, w]].
\end{equation}
By (\ref{Eqn:5.1_0708}) and (\ref{Eqn:5.5_0708}), we have 
\begin{equation*}
\{z_{\gamma\, \lambda} \phi_w\}= \{z_{\gamma\, \lambda }\ w-\sum_{\alpha\in S(1/2)}z^*_\alpha [z_\alpha, w] +\frac{1}{3} \sum_{\alpha, \beta\in S(1/2)} z^*_\alpha z^*_\beta[z_\beta,[z_\alpha, w]]-\sum_{\alpha}k(z_\alpha| w) \partial z^*_\alpha\}=0.
\end{equation*}
Also, by direct computations, we can show that $\{z_{\gamma\ \lambda\ } \phi_f\}=0.$
\end{proof}

\begin{prop} \label{Prop:aff_bracket}
Let $f$ be a minimal nilpotent in $\g$. Let $v,v_1, v_2$ be elements in $\g_f(0)$ and $w,w_1, w_2$ be elements in $\g(-\frac{1}{2}).$
The $\lambda$-brackets between generators in Proposition \ref{Prop:generator} are as follows:
\begin{equation} \label{Eqn:5.7_0708}
\begin{aligned}
& \{ \phi_{v_1\, \lambda} \phi_{v_2}\}= \phi_{[v_1, v_2]} +k\lambda(v_1|v_2),\\
& \{\phi_{v\, \lambda}\phi_w\}=\phi_{[v,w]},\\
& \{\phi_{f\, \lambda}\phi_v\}=-(\partial+\lambda)\phi_v,\\
& \{\phi_{f\, \lambda}\phi_w\}=-\left(\partial+\frac{3}{2} \lambda\right)\phi_w\\
& \{\phi_{f\, \lambda}\phi_f\}=-(\partial+2\lambda)\phi_f,\\
& \{\phi_{w_1\, \lambda}\phi_{w_2}\}=(e|[w_1,w_2])(\phi_f+\sum_{i\in I}\frac{1}{2k}\phi_{a_i}\phi_{b_i}), \\
& +\sum_{\alpha\in S(1/2)}(-1)^{p(w_1)p(w_2)}\phi_{[w_2,z^*_\alpha]^\#}\phi_{[z_\alpha,w_1]^\#}-\sum_{\alpha, \beta\in S(1/2)}k^2\lambda^2(z_\alpha,w_1)(z_\beta,w_2)[z^*_\alpha,z^*_\beta],
\end{aligned}
\end{equation}
where $\{a_i | i\in I\}$ and $\{b_i | i\in I\}$ are bases of $\g_f(0)$ such that $(a_i | b_j)=\delta_{ij}$ and $g^\# \in \g_f(0)$ is the projection of $g\in \g$ onto $\g_f(0).$ 
\end{prop}

\begin{proof}
Let us consider the decomposition $\bigoplus_{i\leq \frac{1}{2}}\g(i)= \g_f\oplus \CC x \oplus \g\left( \frac{1}{2} \right)$ and the differential algebra homomorphism $p:\mathcal{V}(\g,f,k) \to S(\CC[\partial]\otimes \g_f)$ be induced by the projection map $\bigoplus_{i\leq \frac{1}{2}}\g(i)\to \g_f$. Then the map $\iota:=p|_{\WW(\g,f,k)} :\WW(\g,f,k) \to S(\CC[\partial]\otimes \g_f)$ is a differential algebra isomorphism defined by $\phi_u\to u$ for any $u\in \g_f(0)\oplus\g(-\frac{1}{2})$ and $\phi_f\mapsto f-\sum_{i\in I} \frac{1}{2k}a_i b_i$. 
The map $\iota$ naturally induces the one to one correspondence $\iota_\lambda: \WW(\g,f,k)[\lambda] \to S(\CC[\partial]\otimes \g_f)[\lambda].$

It is not hard  to see 
 $\iota_\lambda (\{ \phi_{v_1\, \lambda} \phi_{v_2}\})= [v_1, v_2] +k\lambda(v_1|v_2)$
 and $\iota_{\lambda}^{-1} ([v_1, v_2] +k\lambda(v_1|v_2))=\phi_{ [v_1, v_2]} +k\lambda(v_1|v_2).$
 Hence  $\{ \phi_{v_1\, \lambda} \phi_{v_2}\}=\phi_{ [v_1, v_2]} +k\lambda(v_1|v_2).$ All the equations in   (\ref{Eqn:5.7_0708}) can be proved in similar ways. So we shall show the last one which is most complicate. By taking terms which are not in $\ker \iota_\lambda$, we get 
 \begin{equation} \label{Eqn:5.8_0708}
 \begin{aligned}
 &  \{\phi_{w_1\, \lambda}\phi_{w_2}\}  =  \iota_\lambda^{-1}  \left([w_1,w_2]  -\sum_{\alpha\in S(1/2)}[w_1,z^*_\alpha]^\#[z_\alpha,w_2]^\# \right. \\
 & +\sum_{\alpha\in S(1/2)}(-1)^{p(w_2)(p(\alpha)+p(w_1)}[z^*_\alpha, w_2]^\#[z_\alpha, w_1]^\# \\
& + \sum_{\alpha,\beta\in S(1/2)} (-1)^{p(\beta)(p(\alpha)+p(w_1))}[z^*_\alpha, z^*_\beta][z_\alpha, w_1]^\#[z_\beta, w_2]^\#\\
& +\left. k^2\sum_{\alpha,\beta\in S(1/2)}(z_\alpha| w_1)(z_\beta|w_2)\{\partial z^*_{\alpha\ \lambda} \partial z^*_\beta\}\right).
 \end{aligned}
 \end{equation}
We have 
\begin{equation}\label{Eqn:5.9_0708}
\begin{aligned}
&  \iota_\lambda\left( (e|[w_1,w_2])(\phi_f+\sum_{i\in I}\frac{1}{2k}\phi_{a_i}\phi_{b_i})\right)= [w_1, w_2]; \\
& \iota_\lambda\left(-k^2\lambda^2(z_\alpha|w_1)(z_\beta|w_2)[z^*_\alpha,z^*_\beta]\  \right)=  k^2(z_\alpha| w_1)(z_\beta|w_2)\{\partial z^*_{\alpha\ \lambda} \partial z^*_\beta\}.
\end{aligned}
\end{equation}
Also, we have 
\begin{equation}\label{Eqn:5.10_0708}
  \sum_{\alpha,\beta\in S(1/2)} (-1)^{p(\beta)(p(\alpha)+p(w_1))}[z^*_\alpha, z^*_\beta][z_\alpha, w_1][z_\beta, w_2]=\sum_{\alpha \in S(1/2)} [w_1 , z^*_\alpha][z_\alpha, w_2]
  \end{equation}
 and 
\begin{equation}\label{Eqn:5.11_0708}
\iota_\lambda\left( (-1)^{p(w_1)p(w_2)}\phi_{[w_2,z^*_\alpha]^\#}\phi_{[z_\alpha,w_1]^\#}\right)=(-1)^{p(w_2)(p(\alpha)+p(w_1))}[z^*_\alpha, w_2]^\#[z_\alpha, w_1]^\# .
\end{equation}
By equations (\ref{Eqn:5.8_0708}), (\ref{Eqn:5.9_0708}), (\ref{Eqn:5.10_0708}), (\ref{Eqn:5.11_0708}), we proved our assertion.  
 \end{proof}
 
 Analogously, we can obtain a generating set (as an associative superalgebra) of a finite W-superalgebra associated to a minimal nilpotent and commutators between them. 

\begin{prop} \label{Prop:AA_gen}
Let $f$ be a minimal nilpotent in $\g$. Suppose $v\in \g_f(0)$ and $w\in \g(-\frac{1}{2})$ then the followings are free generators of $W^{fin}(\g,f)$:
\begin{equation}
\begin{aligned}
& \Psi_v= v-\frac{1}{2}\sum_{\alpha\in S(1/2)} z^*_\alpha[z_\alpha, v], \\ 
& \Psi_w=w-\sum_{\alpha\in S(1/2)} z^*_\alpha[z_\alpha, w] \\
& \hskip 10mm+\frac{1}{3} \sum_{\alpha, \beta\in S(1/2)}  \left(\, z^*_\alpha z^*_\beta [z_\beta, [z_\alpha, w]]+(f|[z^*_\alpha, z^*_\beta])[z_\beta,[z_\alpha, w]]\, \right), \\
& \Psi_f= [\text{image of $-\sum_{\alpha \in \bar{S}}u_\alpha u^\alpha$ in $\mathcal{U}(\g,f)$}],   
\end{aligned}
\end{equation}
where $\{u_\alpha\}_{\alpha \in \bar{S}}$ and $\{u_\alpha\}_{\alpha \in \bar{S}}$ are dual bases of $\g$ with respect to $(\, |\, )$. 
\end{prop}

\begin{proof}
It is enough to show that $[z_\gamma, \Psi_v]=[z_\gamma, \Psi_w]=[z_\gamma, \Psi_f]=0$ for any $\gamma \in S(1/2).$ Here we show the most complicate case: $[z_\gamma, \Phi_w]=0.$ We have 
\begin{equation} \label{Eqn:5.13_0726}
\begin{aligned}
&  \sum_{\alpha \in S(1/2)}   [z_\gamma, z^*_\alpha[z_\alpha, w]]= [z_\gamma, w]+\sum_{\alpha \in S(1/2)}(-1)^{p(\alpha)p(\gamma)}[z_\gamma,[z_\alpha, w]] \, ; \\
& \sum_{\alpha, \beta \in S(1/2)}   [z_\gamma, z^*_\alpha z^*_\beta[z_\beta, [z_\alpha, w]]] \\
&=\sum_{\beta \in S(1/2)} z^*_\beta[[z_\beta, z_\gamma],w]  +\sum_{\alpha \in S(1/2)}(-1)^{p(\alpha)p(\gamma)} z^*_\alpha[z_\gamma,[z_\alpha, w]] \\
&  +\sum_{\alpha, \beta \in S(1/2)}(-1)^{p(\gamma)(p(\alpha)+p(\beta))}z^*_\alpha z^*_\beta[z_\gamma, [z_\beta, [z_\alpha, w]]].
\end{aligned}
\end{equation}
The first term and the third term of the second equation in (\ref{Eqn:5.13_0726}) are 
\begin{equation} \label{Eqn:5.14_0726}
\sum_{\beta\in S(1/2)} z^*_\beta[z_\beta,[z_\gamma, w]]=\sum_{\beta\in S(1/2)}z^*_\beta[[z_\beta, z_\gamma], w]+(-1)^{p(\beta)p(\gamma)}z^*_\beta[z_\gamma, [z_\beta, w]] ;
\end{equation}
\begin{equation} \label{Eqn:5.15_0726}
\begin{aligned}
& \sum_{\alpha, \beta\in S(1/2)}(-1)^{p(\gamma)(p(\alpha)+p(\beta))}z^*_\alpha z^*_\beta[z_\gamma,[z_\beta, [z_\alpha, w]]] \\
&= \sum_{\alpha \in S(1/2)} z^*_\alpha[z_\gamma, [z_\alpha , w]]-\sum_{\alpha, \beta\in S(1/2)} (-1)^{p(\gamma)(p(\alpha)+p(\beta))}z^*_\alpha z^*_\beta[[z_\alpha, w], [z_\gamma, z_\beta]] \\
&= \sum_{\alpha \in S(1/2)} z^*_\alpha[z_\gamma, [z_\alpha , w]]-  \sum_{\beta \in S(1/2)}z^*_\beta[[z_\beta, z_\gamma], w] \\
&-\sum_{\alpha, \beta\in S(1/2)} [z^*_\alpha, z^*_\beta][[z_\alpha, w], [z_\gamma, z_\beta]].
\end{aligned}
\end{equation}
Here, from second line to third line, we used $[z_\alpha, z_\beta]=z_\alpha z_\beta-(-1)^{p(\alpha)}z_\beta z_\alpha.$
Also, we have
\begin{equation} \label{Eqn:5.16_0726}
\sum_{\alpha, \beta\in S(1/2)} [z^*_\alpha, z^*_\beta][[z_\alpha, w], [z_\gamma, z_\beta]]=-\sum_{\alpha, \beta\in S(1/2)} [z^*_\alpha, z^*_\beta][z_\gamma,[z_\beta, [z_\alpha, w]]].
\end{equation}
By (\ref{Eqn:5.13_0726}),(\ref{Eqn:5.14_0726}),(\ref{Eqn:5.15_0726}),(\ref{Eqn:5.16_0726}), we conclude that $[z_\gamma, \Psi_w]=0$ for any $\gamma\in S(1/2).$
\end{proof}

\begin{prop}
Let $f$ be a minimal nilpotent in $\g$. Let $v,v_1, v_2$ be elements in $\g_f(0)$ and $w,w_1, w_2$ be elements in $\g(-\frac{1}{2}).$ Let $z_{w}=\sum_{\alpha, \beta\in S(1/2)}(f|[z^*_\alpha, z^*_\beta])[z_\beta, [z_\alpha, w]].$
The commutators between generators in Proposition \ref{Prop:AA_gen} are as follows:
\begin{equation} \label{Eqn:4.17_0727}
\begin{aligned}
& [ \Psi_{v_1}, \Psi_{v_2} ]= \Psi_{[v_1, v_2]},\\
& [ \Psi_v, \Psi_w ]=\Psi_{[v,w]},\\
& [ \Psi_f, \WW^{fin}(\g,f) ]=0,\\
& [ \Psi_{w_1},\Psi_{w_2} ]=(e|[w_1,w_2])(\Psi_f+\sum_{i\in I}\frac{1}{2k}\Psi_{a_i}\Psi_{b_i}), \\
& +\sum_{\alpha\in S(1/2)}(-1)^{p(w_1)p(w_2)}\Psi_{[w_2,z^*_\alpha]^\#}\Psi_{[z_\alpha,w_1]^\#}-(f|[z_{w_1},z_{w_2}]),
\end{aligned}
\end{equation}
where $\{a_i | i\in I\}$ and $\{b_i | i\in I\}$ are bases of $\g_f(0)$ such that $(a_i | b_j)=\delta_{ij}$ and $g^\# \in \g_f(0)$ is the projection of $g\in \g$ onto $\g_f(0).$ 
\end{prop}

\begin{proof}
The argument in the proof of Proposition \ref{Prop:aff_bracket} works.
\end{proof}

Also for finite classical W-superalgebras, we can obtain similar propositions.  

\begin{prop} \label{Prop:fin_generator}
Let $f$ be a minimal nilpotent in $\g$. Suppose $v\in \g_f(0)$ and $w\in \g(-\frac{1}{2})$ then the followings are free generators of $\WW^{fin}(\g,f)$:
\begin{equation}
\begin{aligned}
& \psi_v= v-\frac{1}{2}\sum_{\alpha\in S(1/2)} z^*_\alpha[z_\alpha, v], \\ 
& \psi_w=w-\sum_{\alpha\in S(1/2)} z^*_\alpha[z_\alpha, w]+\frac{1}{3} \sum_{\alpha, \beta\in S(1/2)} z^*_\alpha z^*_\beta[z_\beta, [z_\alpha, w]], \\
& \psi_f= [\text{image of $-\sum_{\alpha \in \bar{S}}u_\alpha u^\alpha$ in $\mathcal{S}(\g,f)$}],   
\end{aligned}
\end{equation}
where $\{u_\alpha\}_{\alpha \in \bar{S}}$ and $\{u_\alpha\}_{\alpha \in \bar{S}}$ are dual bases of $\g$ with respect to $(\, |\, ).$
\end{prop}

\begin{proof}
By Corollary \ref{Cor:affine-finite}, we obtain free generators of $\WW(\g,f)$ from the generators of $\WW(\g,f,k).$
\end{proof}

Similarly, we obtain the Poisson brackets between generating elements.

\begin{prop}
Let $f$ be a minimal nilpotent in $\g$. Let $v,v_1, v_2$ be elements in $\g_f(0)$ and $w,w_1, w_2$ be elements in $\g(-\frac{1}{2}).$
The Poisson brackets between generators in Proposition \ref{Prop:fin_generator} are as follows:
\begin{equation} \label{Eqn:5.7_0708}
\begin{aligned}
& \{ \psi_{v_1}, \psi_{v_2}\}= \psi_{[v_1, v_2]},\\
& \{\psi_v, \psi_w\}=\psi_{[v,w]},\\
& \{\psi_f, \WW^{fin}(\g,f)\}=0,\\
& \{\psi_{w_1},\psi_{w_2}\}=(e|[w_1,w_2])(\psi_f+\sum_{i\in I}\frac{1}{2k}\psi_{a_i}\psi_{b_i}), \\
& +\sum_{\alpha\in S(1/2)}(-1)^{p(w_1)p(w_2)}\psi_{[w_2,z^*_\alpha]^\#}\psi_{[z_\alpha,w_1]^\#},
\end{aligned}
\end{equation}
where $\{a_i | i\in I\}$ and $\{b_i | i\in I\}$ are bases of $\g_f(0)$ such that $(a_i | b_j)=\delta_{ij}$ and $g^\# \in \g_f(0)$ is the projection of $g\in \g$ onto $\g_f(0).$ 
\end{prop}

\begin{proof}
By Corollary \ref{Cor:affine-finite} and Proposition \ref{Prop:aff_bracket}, we can prove our assertion.
\end{proof}

\subsection{Examples }\ \\

\begin{ex}
Let $\g=spo(2|1)\subset \mathfrak{gl}(2|1).$ Then the even part $\g_{\bar{0}}$ is generated by an $\sll_2$-triple $(e_{ev}, h, f_{ev})$ and the odd part  $\g_{\bar{1}}$ is generated by $e_{od}$ and $f_{od}$ such that 
\begin{equation*} 
h= \left(
\begin{array}{ccc}
1& 0& 0\\
0 &-1 &0 \\
0 & 0 & 0
\end{array}\right)
\qquad
e_{ev}= \left(
\begin{array}{ccc}
0 & 1& 0\\
0 & 0 &0 \\
0 & 0 & 0
\end{array}\right)
\qquad
f_{ev}= \left(
\begin{array}{ccc}
0 & 0 & 0\\
1 & 0 &0 \\
0 & 0 & 0
\end{array}\right)
\end{equation*}

\begin{equation*}
e_{od}= \left(
\begin{array}{ccc}
0 & 0 & 1\\
0 & 0 &0 \\
0 & 1 & 0
\end{array}\right)
\qquad
f_{od}= \left(
\begin{array}{ccc}
0 & 0 & 0\\
0 & 0 &1 \\
-1 & 0 & 0
\end{array}\right) .
\end{equation*} 

Then we have 
\[ [h,e_{od}]=e_{od}, \quad [h,f_{od}]=-f_{od}, \quad [e_{od}, f_{od}]=[f_{od}, e_{od}]=-h \]
\[ [e_{od}, f_{ev}]=-f_{od}, \quad [f_{od}, e_{ev}]=-e_{od}, \quad [e_{od}, e_{od}]=2 e_{ev}, \quad [f_{od}, f_{od}]=-2f_{ev}.\]
Note that $f_{od}\in \g(-1/2)$, $f_{ev}\in  \g(-1)$,  $e_{od}\in \g(1/2)$, $e_{ev}  \in \g(1)$

Consider the supersymmetric invariant bilinear form $(\, | \, )$ such that $(h|h)=2(e_{ev}|f_{ev})=2$ and $(e_{od}|f_{od})=-2.$ 

In order to find free generators of $\WW(\g,f_{ev},1)$, we want to find an element  $X\in S(\CC[\partial]\otimes \g)/I$ satisfying $\ad_\lambda e_{od}  (X) =0+I.$ Here we recall that $I$ is the differential algebra ideal generated by $e_{ev}+1$. Note that  $\ad_\lambda e_{od} (X) =0+I$ implies $\ad_\lambda e_{ev}  (X)=0+I.$ It is not hard to find two elements 
\[\phi_{od}:=f_{od}-\frac{1}{2}e_{od}h-\partial e_{od}, \quad \phi_{ev}:=f_{ev}+\frac{1}{2}f_{od}e_{od}-\frac{1}{4}h^2+\frac{1}{4}e_{od}\partial e_{od}-\frac{1}{2}\partial h\] which satisfy 
\[ \ad_\lambda e_{od} (\phi_{od})=0+I ,\qquad  \ad_\lambda e_{od} (\phi_{ev})=0+I .\]
Hence  
\[ \WW(\g,f_{ev},1)=S(\CC[\partial]\otimes (\CC\phi_{od}\oplus \CC\phi_{ev})).\]
 as a differential algebra. 
By direct computations, we can check that the $\lambda$-bracket of $\WW(\g, f_{ev},1)$ is defined by 
\begin{equation*}
\begin{aligned}
& \{\phi_{od}\, _\lambda\, \phi_{od}\}=-2 \phi_{ev}-2\lambda^2  \\
& \{\phi_{ev}\, _\lambda\, \phi_{od}\}= -(\partial+\frac{3}{2}\lambda)\phi_{od}\\
&\{\phi_{ev}\, _\lambda\, \phi_{ev}\}=-(\partial+2\lambda)\phi_{ev}-\frac{1}{2}\lambda^3.
\end{aligned}
\end{equation*}

\end{ex}

\begin{ex}
Let $\g=spo(2|3)\subset \mathfrak{gl}(2|3)$. As a matrix form  
\[ \mathfrak{gl}(2|3) = \left(
\begin{array}{cc}
A & B \\
C & D
\end{array}\right) 
\]
where $A,B,C,D$ are $2\times 2$, $2\times 3$, $3\times 2$, $3\times 3$ matrices, respectively.  We denote by $e_{ij}, e_{i\bar{j}}, e_{\bar{i}j}, e_{\bar{i}\bar{j}} \in \mathfrak{gl}(2|3)$ the matrix with $1$ in $ij$-entry of $A,B,C,D$, respectively,  and $0$ in other entries.

Consider the $\sll_2$-triple $(e,h,f)$ where 
\[ h=e_{11}-e_{22}+2(e_{\bar{1}\bar{1}}-e_{\bar{2}\bar{2}}), \ e=e_{12}+e_{\bar{1}\bar{3}}-e_{\bar{3}\bar{2}}, \ f=e_{21}+2e_{\bar{3}\bar{1}}-2e_{\bar{2}\bar{3}}.\]
Take the supersymmetric invariant bilinear form $( \, |\, )$ such that $(e|f)=\frac{1}{2}(h|h)=1.$ Then $\g$ is generated by the following elements
  \begin{equation*}
 \begin{aligned}
 & H_1=e_{11}-e_{22}, \quad H_2=e_{\bar{1}\bar{1}}-e_{\bar{2}\bar{2}}, \\
 & E_{11}= e_{\bar{1}1}-e_{2\bar{2}}, \quad E_{12}=e_{\bar{3}2}+e_{1\bar{3}}, \quad E_{21}=e_{12},\quad E_{22}=e_{\bar{1}\bar{3}}-e_{\bar{3}\bar{2}}, \quad E_3=e_{\bar{1}2}+e_{1\bar{2}} \\
 & F_{11}=e_{1\bar{1}}+e_{\bar{2}2}, \quad F_{12}=e_{\bar{3}1}-e_{2\bar{3}}, \quad F_{21}=e_{21}, \quad F_{22}=e_{\bar{2}\bar{3}}-e_{\bar{3}\bar{1}}, \quad F_3=e_{\bar{2}1}-e_{2 \bar{1}}.
 \end{aligned} 
 \end{equation*}
Note that $H_1, H_2\in \g(0)$, $E_{11}, E_{12}\in \g(1/2)$, $E_{21}, E_{22}\in \g(1)$, $E_{3}\in \g(3/2)$ and $F_{11}, F_{12}\in \g(-1/2)$, $F_{21}, E_{22}\in \g(-1)$, $F_{3}\in \g(-3/2).$
In the differential algebra  $S(\CC[\partial]\otimes \g)/I$ where $I$ is the differential algebra ideal generated by $m+(f|m)$ for $m\in \m,$ we have $e=-1$, $E_{21}=\frac{1}{3}$ and $E_{22}=-\frac{4}{3}.$ Moreover, we have 
\[ [E_{12}, E_{12}]=\frac{2}{3}, \qquad [E_{11}, E_{12}]=-\frac{4}{3}, \qquad [E_{11}, E_{11}]=0,\]
\[ (E_{11} |F_{11})=\frac{2}{3}, \qquad (E_{12} |F_{12})=-\frac{2}{3}.\]

Since $\ker(\ad f)\subset spo(2|3)$ is generated by four elements $F_{11}-\frac{1}{2} F_{12}$, $F_{21}$, $F_{22}$ and $F_{23}, $ the W-algebra $\WW(\g, f,1)$ is freely generated by four elements as a differential algebra. We can see that the following four elements
\begin{equation*}
\begin{aligned}
 \phi_1\  &= F_{11}-\frac{1}{2}F_{12}+\frac{3}{4}H_1E_{12}+\frac{3}{4}H_2E_{12}+\frac{3}{8}H_2 E_{11}+\frac{1}{2}\partial E_{11}+\frac{1}{2}\partial E_{12};  \\
 \phi_{21}& =F_{21}-\frac{3}{4}F_{12}E_{11}-\frac{9}{8}H_1E_{11}E_{12}+\frac{3}{4}H_1^2 \\
 &  -\frac{3}{8}E_{12}\partial E_{11}+\frac{3}{8}E_{11} \partial E_{12} -\frac{3}{16} E_{11}\partial E_{11}-\frac{1}{2}\partial H_1; \\
 \phi_{22}& = F_{22}+\frac{3}{4}F_{11}E_{11}-\frac{3}{4}F_{12}E_{12}-\frac{3}{8}F_{12}E_{11}-\frac{9}{16}H_1E_{11}E_{12}+\frac{3}{8}H_2^2\\
 & -\frac{3}{8}E_{12}\partial E_{11}-\frac{3}{16}E_{11}\partial E_{11}+\frac{1}{2} \partial H_2; \\
 \phi_3 \ & = F_{3}-\frac{3}{2} F_{21}E_{12}+\frac{3}{4}F_{22}E_{11}-\frac{3}{4}F_{21}E_{11}+\frac{9}{8}F_{12}E_{11}E_{12}-\frac{9}{16}H_1^2E_{11}+\frac{3}{4}H_2F_{12} \\
 & -\frac{3}{2}H_1F_{11}+\frac{9}{8}F_{11}E_{11}E_{12}-\frac{9}{8}H_1^2E_{12}-\frac{9}{8}H_1H_2E_{12} \\
 & -\frac{3}{8}H_2\partial E_{11}-\frac{3}{4}H_1 \partial E_{12}+\frac{9}{16}E_{11}E_{12}\partial E_{12} +\frac{3}{8}\partial H_1 E_{11}+\frac{1}{2} \partial F_{12}-\frac{1}{4} \partial^2 E_{11}
\end{aligned}
\end{equation*}
freely generate $\WW(\g,f,1).$ The $\lambda$-bracket is defined by 
\begin{equation*}
\begin{aligned}
&  \{\phi_{1\, \lambda} \phi_1\}=  \phi_{22}-\frac{1}{2}\phi_{21}+\frac{1}{6}\lambda^2; \\
&  \{\phi_{1\, \lambda} \phi_{21}\}=  \phi_3+\frac{1}{2}\partial\phi_1+\frac{1}{2}\lambda\phi_1; \\
&  \{\phi_{1\, \lambda} \phi_{22}\}= \frac{1}{2}\phi_3+\frac{1}{2}\partial \phi_1+ \lambda \phi_1;  \\
&  \{\phi_{1\, \lambda} \phi_{3}\}=  \frac{1}{2}\partial \phi_{21}+\frac{1}{2}\partial \phi_{22}+ \frac{3}{2} \lambda\phi_{21}+\frac{1}{6}\lambda^3; \\
&  \{\phi_{21 \, \lambda} \phi_{21}\}= -(\partial+2\lambda)\phi_{21}+\frac{1}{6}\lambda^3; \\
&  \{\phi_{21\, \lambda} \phi_{22}\}=  0; \\
&  \{\phi_{21\, \lambda} \phi_3\}=   \frac{3}{2}\phi_{21}\phi_1-\frac{1}{2}\lambda\phi_3+\frac{1}{2}\lambda^2\phi_1;\\
&  \{\phi_{22\, \lambda} \phi_{22}\}= \frac{1}{2}(\partial+2\lambda)\phi_{22}-\frac{1}{6}\lambda^3 ; \\
&  \{\phi_{22\, \lambda} \phi_3\}= -\frac{9}{2}\phi_{1}\phi_{21}+(\lambda+\frac{1}{2}\partial)\phi_3;  \\
&  \{\phi_{3 \, \lambda} \phi_3\}=   -3\phi_{21}\phi_{22}-3\phi_3\phi_1 + \left(\frac{3}{2}\lambda^2+\frac{3}{2}\lambda\partial+\frac{1}{2}\partial^2\right)\phi_{21}.\\
\end{aligned}
\end{equation*} 
The way we get Poisson $\lambda$-brackets is same as the argument in the proof of Proposition \ref{Prop:aff_bracket}.
\end{ex}

\section{Affine classical fractional W-(super)algebras} \label{Sec:frac}

Recall that $\g$ is a simple Lie superalgebra with an $\sll_2$-triple $(e,2x,f)$ and the supersymmetric bilinear invariant form $(\, | \, )$ such that $(e|f)=2(x|x)=1$.  Supppose $\g[z,z^{-1}]=\CC[z,z^{-1}]\otimes \g$ and $\g[z]:= \CC[z]\otimes \g$ are  Lie superalgebras with the bracket 
\[ [a z^m, b z^n]:=[a,b] z^{m+n}, \quad m, n \in \ZZ_, \, a,b\in \g.\] 
The Lie superalgebra $\g[z,z^{-1}]$ has the bilinear invariant form $(\, |\, )$ which is induced by that of $\g$, i.e.
\[ (a z^m | b z^n)=(a|b)\delta_{m+n, 0}, \quad a,b\in \g.\]
Denote
\[\g^{[t]}=\g[z]/z^{t+1} \g[z], \quad t\in \ZZ_{>0}\]
and let $\mathcal{V}_t$ be the differential algebra defined by 
\[ \mathcal{V}_t(\g,f,k)= S(\CC[\partial]\otimes \g^{[t]})/I^t \]
  for the ideal  $I^t$  generated by $\{ z^t m+(f|m)|m\in \m\}.$ Note that if $t=0$ then $\mathcal{V}_t(\g,f,k)=\mathcal{V}(\g,f,k)$ which appears when we define the ordinary classical affine W-superalgebra.

In order to introduce another description of $\mathcal{V}_t(\g,f,k)$, define a gradation on $\g[z]$ by letting
\[ \gr(z)=1+d, \quad \gr(g)=j, \text{ for } g\in \g(j).\]
Here $d$ satisfies that $\g=\bigoplus_{i=-d}^{d} \g(i)$ and $\g(d)\neq 0.$ We denote
\[ \g_l:= \{g\in \g[z]| \gr(g)=l\}, \text{ for } l\geq -d.\]
If $p\in \g(d)$ then  
\[ \gr(f z^t)=\gr(pz^{t-1})=t(1+d)-1\]
and
\[\g_{[t,d]}:=\g[z]\,  \big/ \bigoplus_{t'>(1+d)t+1}\g_{t'}\simeq \bigoplus_{t'\leq(1+d)t+1}\g_{t'} \ \simeq\  \g^{[t]} \oplus z^{t+1} \g(-d)\] 
as vector superspaces.

 \begin{defn}\label{Prop:V_frac}
The differential algebra $ S(\CC[\partial]\otimes \g_{[t,d]})$
has a PVA structure  induced by the  Poisson $\lambda $-bracket on $\CC[\partial]\otimes\g_{[t,d]}:$
\begin{equation}
\{ a z^p_{\ \lambda} b z^q\}= 
\left\{ \begin{array}{ll}
 [a,b]+k\lambda(a|b) & \text{ if } p=q=0, \\
\ \  0 & \text{ if only one of } p,q \text{ is } 0, \\
\,- [a,b] z^{m+n} & \text{ if } p,q \neq 0.
\end{array}
\right.
\end{equation} 
for $a,b\in \g$, $p,q\in \ZZ_{\geq 0}$ and $k\in \CC.$
\end{defn}

Moreover, we have the differential algebra isomorphism 
\[ \mathcal{V}_t(\g,f,k)\simeq  S(\CC[\partial]\otimes \g_{[t,d]})/I_{[t,d]},\]
where $I_{[t,d]}$ is the ideal generated by 
\[\{ m+(m|\Lambda_t)|  m\in \g_{(1+d)+1}\} \text{ for } \Lambda_t=f z^{-t}+p z^{-t-1}.\]  Here $(m|\Lambda_t)\in \CC$ is induced from the bilinear form of $\g[z,z^{-1}].$

Recall that $\n= \bigoplus_{i>0} \g(i)\subset \g.$ Define the $\ad_\lambda \n$ action on $\mathcal{V}_t(\g,f,k)$  induced by the action on $S(\CC[\partial]\otimes\g[z])$:
\begin{equation} \label{Eqn:ad action}
\begin{aligned}
&  \ad_\lambda n (az^p)= [n,a]+\delta_{p,0} k\lambda(n|a), \\
&  \ad_\lambda n (AB)= (-1)^{p(n)p(A)} A\ \ad_\lambda n (B) + \ad_\lambda n (A) B, \\
& \ad_\lambda n (\partial A)= (\partial +\lambda) \ad_\lambda n (A),
\end{aligned}
\end{equation}
for $n\in \n$, $a \in \g$, $A,B\in S(\CC[\partial]\otimes \g[z])$, $p\in \ZZ_{\geq 0}.$

Take the superspace
\begin{equation}
\WW_t(\g,f,k)= \mathcal{V}_t(\g,f,t)^{\ad_\lambda \n}=\{ A\in \mathcal{V}_t(\g,f,k)| \ad_\lambda n (A)=0 \text{ for any } n \in \n\}.
\end{equation}
By the definition of $\ad_\lambda \n$ action (\ref{Eqn:ad action}) guarantees that $\WW_t(\g,f,k)$ is a differential algebra. Moreover, we have the following proposition.

\begin{prop}
The differential algebra $\WW_t(\g,f,k)$ is a PVA endowed with the $\lambda$-bracket induced from that of $S(\CC[\partial]\otimes \g_{[t,d]})$ in Definition \ref{Prop:V_frac}.
\end{prop}

\begin{proof}
Let us first show when $\g$ is a Lie algebra. Suppose $A=\sum_{i\in T} A_i^0 A_i^>$ and $B=\sum_{j\in T'} B_j^0 B_j^>$ are elements of $\WW_t(\g,f,k)$ for $A_i^0, B_j^0\in S(\CC[\partial]\otimes \g)$ and $A_i^>, B_j^> \in S(\CC[\partial]\otimes z\g[z])$. We need to show that 
\[ \{A_\lambda B\} \in \WW_t(\g,f,k).\]
Note that, $\text{ for } n\in \n$,
\begin{enumerate}[(i)]
\item $ \ad_\lambda n (A)= \sum_{i\in T} A_i^> \cdot \ad_\lambda n(A_i^0)+ A_i^0 \cdot \ad _\lambda n (A_i^>)=0$;
\item $ \ad_\lambda n (B)= \sum_{j\in T'} B_j^> \cdot \ad_\lambda n(B_j^0)+ B_j^0 \cdot \ad _\lambda n (B_j^>)=0$;
\item 
$ \ad_\lambda n \big( S(\CC[\partial]\otimes \g z^k )\big)\subset  S(\CC[\partial]\otimes \g z^k )[\lambda].$
\end{enumerate}
Since we have
\[ \{A_\lambda B\}= \sum_{i\in T, j\in T'} B_j^>\{A^0_{i\, \lambda+\partial } B^0_j\}_\to A^>_i + B_j^0\{A^>_{i\, \lambda+\partial } B^>_j\}_\to A^0_i, \]
 it is enough to show that 
\begin{equation} \label{Eqn:6.4_150830}
 \sum_{i\in T, j\in T'} \ad_\mu n \big(  B_j^>\{A^0_{i\, \lambda+\partial } B^0_j\}_\to A^>_i  \big)= 0= \sum_{i\in T, j\in T'}  \ad_\mu n   \big(  B_j^0\{A^>_{i\, \lambda+\partial } B^>_j\}_\to A^0_i \big).
\end{equation}
In order to show the first equality in (\ref{Eqn:6.4_150830}), we expand 
\begin{align}
 \ad_\mu n \big(  B_j^> \{A^0_{i\, \lambda+\partial } B^0_j\}_\to  A^>_i  \big)
& = \big( \{A^0_{i\, \lambda+\partial} B^0_j\}_\to  A_i^>\big)  \, \ad_\mu n (B_j^>) \label{Eqn:6.5_150830}\\
& +\sum_{k\geq 0} B_j^> \left( \frac{(\lambda+\partial)^k}{k!}A_i^> \right) \, \ad_\mu n (A^0_{i \, (k)}B^0_j) \label{Eqn:6.6_150830}\\
& + \sum_{k\geq 0} B_j^>\,  A^0_{i \, (k)}B^0_j\, \left( \frac{(\lambda+\mu+ \partial)^k}{k!}A_i^> \right) \, \ad_\mu n ( A^>_i  ). \label{Eqn:6.7_150830}
\end{align}
Since
\[\frac{1}{k!} \ad_\mu n (A^0_{i \ (k) }B^0_j) = \sum_{k'\geq k} \frac{k'!}{k!(k'-k)!}\big[\ad_\mu n (A_i^0)\big]_{(k')} B^0_j \ \mu^{k'-k} +\frac{1}{k!} A^0_{i\ (k) } \big[ \ad_\mu n (B^0_j)\big]\]
we have
\begin{align}
  (\ref{Eqn:6.6_150830}) = & \sum_{k\geq 0} B^>_j \left( \frac{(\lambda+\partial)^k}{k!} A^>_i \right) \sum_{k'\geq k} \frac{\mu^{k'-k}}{(k'-k)!}\big[\ad_\mu n(A_i^0)\big]_{(k')}B^0_j  \label{Eqn:6.8_150830} \\
+ &  \sum_{k\geq 0} B^>_j \left( \frac{(\lambda+\partial)^k}{k!} A^>_i \right) A^0_{i\, (k)}\big[\ad_\mu n (B_j^0)\big] .\label{Eqn:6.9_150830}
\end{align}

By (i) and (iii), 
\begin{equation} \label{Eqn:6.10_150830}
\begin{aligned}
 \sum_{i\in T, j\in T'} (\ref{Eqn:6.8_150830}) & = - \sum_{i\in T, j\in T'} \left[ \sum_{k\geq 0} B^>_j \left( \frac{(\lambda+\partial)^k}{k!} \ad_\mu n (A^>_i) \right) \sum_{k'\geq k} \frac{\mu^{k'-k}}{(k'-k)!} A^0_{i\, (k')}B^0_j \right]  \\
 & =  -\sum_{i\in T, j\in T'} (\ref{Eqn:6.7_150830}).
\end{aligned}
\end{equation}
Also, by (ii) and (iii), 
\begin{equation} \label{Eqn:6.11_150830}
\begin{aligned}
\sum_{i\in T, j\in T'} (\ref{Eqn:6.9_150830}) = & -\sum_{i\in T, j\in T'} \left[   \sum_{k\geq 0} \ad_\mu n (B^>_j) \left( \frac{(\lambda+\partial)^k}{k!} A^>_i \right) A^0_{i\, (k)} B_j^0    \right]  \\
& = - \big( \{A^0_{i\, \lambda+\partial} B^0_i\}_\to  A_i^>\big)  \, \ad_\mu n (B_i^>).
\end{aligned}
\end{equation}
Hence, by (\ref{Eqn:6.4_150830})- (\ref{Eqn:6.11_150830}), we have
\[  \sum_{i\in T, j\in T'}\ad_\mu n (B_j^>\{A^0_{i\, \lambda+\partial } B^0_j\}_\to A^>_i )=0\]
for any $n\in\n.$ Using the same argument, we have 
\[  \sum_{i\in T, j\in T'}  \ad_\mu n   \big(  B_j^0\{A^>_{i\, \lambda+\partial } B^>_j\}_\to A^0_i \big)=0.\]
Thus \[ \ad_\mu n (\{A_\lambda B\})=0, \text{ that is } \{A_\lambda B\}\in \WW_t(\g,f,k)[\lambda]\]  and $\WW_t(\g,f,k)$ is endowed with the PVA structure.

If $\g$ is a Lie superalgebra having an odd part, we can prove the proposition similarly. The only thing we need is considering change of signs according to the supersymmetry of $\WW_t(\g,f,k)$. So we proved our assertion.
\end{proof}

\begin{defn}
The PVA  
\[ \WW_t(\g,f,k)= \mathcal{V}_t(\g,f,k)^{\ad_\lambda \n}\]
 endowed with the Poisson $\lambda$-bracket induced by that of $S(\CC[\partial]\otimes \g_{[t,d]})$ in Definition  \ref{Prop:V_frac} is called the {\it $t$-th classical affine W-algebra} associated to $\g$, $f$ and $k\in \CC.$
\end{defn}

Now we assume that $f$ is a minimal nilpotent in $\g$ and $\Lambda_t=f z^{-t}+e z^{-t-1}$. Then we can find generators of $\WW_t(\g,f,k)$

\begin{prop} \label{Prop: frac_gen}
Let $u\in \g\big(\frac{1}{2}\big)$, $v\in \g(0)$ and $w\in \g\big(-\frac{1}{2}\big)$ and recall $\{z_\alpha|\alpha\in S(1/2)\}$ and $\{z^*_\alpha|\alpha\in S(1/2)\}$ are bases of $\g\big(\frac{1}{2}\big)$ such that $[z_\alpha, z^*_\beta]=-e.$ For $p=0, \cdots, t-1$, let $\eta_t$ be the linear map defined as follows. Then the elements listed below are generators of $\WW_t(\g,f,k).$
\begin{equation}
\begin{aligned}
 \eta_t(e z^p) & =e z^p;\\
 \eta_t(u z^p) & =u z^p-\sum_{\alpha \in S(1/2)}z^*_\alpha z^t[z_\alpha, uz^p] ;\\
 \eta_t(v z^p) & = v z^p-\sum_{\alpha \in S(1/2)}z^*_\alpha z^t[z_\alpha, vz^p]-xz^t [e,v z^p] \\
& +\frac{1}{2} \sum_{\alpha, \beta\in S(1/2)} z^*_\alpha z^t \,  z^*_\beta z^t [z_\beta, [z_\alpha, v z^p]];\\
 \eta_t(w z^p) & = w z^p -\sum_{\alpha\in S(1/2)} z^*_\alpha z^t[z_\alpha, wz^p]-xz^t [e,w z^p]  \\
 & +\frac{1}{2} \sum_{\alpha, \beta\in S(1/2)} z^*_\alpha z^t \,  z^*_\beta z^t [z_\beta, [z_\alpha, w z^p]]+x z^t \,  z^*_\beta z^t [z_\beta, [e, w z^p]]\\
 & -\frac{1}{6}\sum_{\alpha, \beta,\gamma\in S(1/2)} z^*_\alpha z^t \,  z^*_\beta z^t \,  z^*_\gamma z^t [z_\gamma, [z_\beta, [z_\alpha, w z^p ]]]-k\, \delta_{p,0}\sum_{\alpha\in S(1/2)} \partial z^*_\alpha z^t (z_\alpha|w).
\end{aligned}
\end{equation}
Denote by $\eta'_t(f z^p)$ the element in $\mathcal{V}_t(\g,f,k)$ which is obtained from $\eta_t(w z^p)$ by substituting $w z^p$ by $f z^p$ and $w$ by $f$.
\begin{equation}
\begin{aligned}
\eta_t(f z^p) & = \eta'_t(f z^p) - (x z^t)^2(ez^p)- \sum_{\alpha, \beta\in S(1/2)} x z^t\, z^*_\alpha z^t \, z^*_\beta z^t [z_\beta, [z_\alpha, x z^p]]  \\
& +\frac{1}{24}  \sum_{\alpha, \beta,\gamma,\delta \in S(1/2)} z^*_\alpha z^t \,  z^*_\beta z^t \,  z^*_\gamma z^t\,  z^*_\delta z^t [z_\delta, [z_\gamma, [z_\beta, [z_\alpha, f z^p ]]]] \\
& -k \, \delta_{p,0} \big( \partial x z^t   - \frac{1}{2}\sum_{\alpha\in S(1/2)}(\partial z^*_\alpha z^t)\, z_\alpha z^t\big)
\end{aligned}
\end{equation}
Also, for  $v\in \g_f(0)$ and $w\in \g\big(-\frac{1}{2}\big)$,  the followings are generators of $\WW_t(\g,f,k).$
\begin{equation}
\begin{aligned}
\eta_t(v z^t) & =  v z^t-\frac{1}{2}\sum_{\alpha \in S(1/2)}z^*_\alpha z^t[z_\alpha, vz^t] ;\\
\eta_t(w z^t) & = w z^t  -\sum_{\alpha\in S(1/2)} z^*_\alpha z^t [z_\alpha, w z^t]+\frac{1}{3} \sum_{\alpha, \beta\in S(1/2)} z^*_\alpha z^t\,  z^*_\beta z^t [z_\beta, [z_\alpha, w z^t]] ; 
\end{aligned}
\end{equation}
Also, we take $\eta_t(f z^t)$ by substituting $f z^p$ in $\eta_t(f z^p)$ with $f z^t$.
\end{prop}

\begin{proof}
It is enough to show that $\ad_\lambda n (\eta_t(g))=0$ for any $n\in \n$ and $g\in \bigoplus_{i=0}^{t-1} \g[z] \oplus \g_f z^t.$ Here we show $\ad_\lambda n (\eta_t(w z^p))=0$ by direct computations. Other cases also can be proved similarly. Note that
\[ [z_\beta, z^*_\alpha z^t]= \delta_{\alpha,\beta} \in \mathcal{V}_t(\g,f,k).\]
Hence, for any $\delta\in S(1/2),$
\begin{equation*}
\begin{aligned}
& (i) \sum_{\alpha \in S(1/2)}\ad_\lambda z_\delta \big( z^*_\alpha[z_\alpha,w z^p] \big) = [z_\delta, w z^p]+\sum_{\alpha\in S(1/2)} (-1)^{p(\alpha)p(\delta)}z^*_\alpha z^t \ [z_\delta, [z_\alpha, w z^p]]; \\
& (ii)\quad  \ad_\lambda z_\delta (x z^t [e, w z^p]) = -\frac{1}{2} z_\delta z^t [e, w z^p] +x z^t [z_\delta, [e, w z^p]] ;\\
& (iii)\sum_{\alpha, \beta\in S(1/2)} \ad_\lambda z_\delta \big(  z^*_\alpha z^t \, z^*_\beta z^t  [z_\beta, [z_\alpha, w z^p]]\big) = \sum_{\beta\in S(1/2)} z^*_\beta z^t [z_\beta, [z_\delta, w z^p]]  \\
& + \sum_{\alpha \in S(1/2)} (-1)^{p(\alpha)p(\delta)} z^*_\alpha z^t [z_\delta, [z_\alpha, w z^p]
+ \sum_{\alpha, \beta\in S(1/2)} (-1)^{p(\delta)(p(\alpha)+p(\beta))} z^*_\alpha z^t \, z^*_\beta z^t [z_\delta, [z_\beta, [z_\alpha, w z^p]]]
\\
& = \sum_{\beta\in S(1/2)}(-1)^{p(\beta)p(\delta)} z^*_\beta z^t [z_\delta, [z_\beta, w z^p]]- z_\delta z^t [e, w z^p] \\
& + \sum_{\alpha \in S(1/2)} (-1)^{p(\alpha)p(\delta)} z^*_\alpha z^t [z_\delta, [z_\alpha, w z^p]
+  \sum_{\alpha, \beta\in S(1/2)} (-1)^{p(\delta)(p(\alpha)+p(\beta))} z^*_\alpha z^t \, z^*_\beta z^t [z_\delta, [z_\beta, [z_\alpha, w z^p]]].
\end{aligned}
\end{equation*}
The last equality holds by the fact that 
\begin{equation*}
\begin{aligned}
& \sum_{\beta\in S(1/2)} z^*_\beta z^t [[z_\beta, z_\delta], w z^p] =  \sum_{\beta\in S(1/2)} z^*_\beta z^t [([z_\beta, z_\delta]|f) e, w z^p] = -z_\delta z^t [e, w z^p].
\end{aligned}
\end{equation*}
Also, we have
\begin{equation*}
\begin{aligned}
& (iv)  \sum_{\beta\in S(1/2)} \ad_\lambda z_\delta \big( xz^t\, z^*_\beta z^t [z_\beta, [e, w z^p]] \big)\\
& = -\frac{1}{2} \sum_{\beta\in S(1/2)}z_\delta z^t \, z^*_\beta z^t [z_\beta, [e, w z^p]] +x z^t [z_\delta, [e, w z^p]]; \\
& (v)  \sum_{\alpha, \beta, \gamma\in S(1/2)} \ad_\lambda z_\delta \big( z^*_\alpha z^t\, z^*_\beta z^t\, z^*_\gamma z^t [z_\gamma, [z_\beta, [z_\alpha, w z^p]]]\big) \\
& = 3\sum_{\alpha, \beta\in S(1/2)} (-1)^{p(\delta)(p(\alpha)+p(\beta))} z^*_\alpha z^t \, z^*_\beta z^t [z_\delta, [z_\beta, [z_\alpha, w z^p]]]-3\sum_{\alpha\in S(1/2)}z_\delta z^t \, z^*_\alpha z^t [z_\alpha, [e, w z^p]]; \\
& (vi) \sum_{\alpha\in S(1/2)} \ad_\lambda z_\delta \big( \partial z^*_\alpha z^t (z_\alpha|w) \big) = \lambda (z_\delta|w)
\end{aligned}
\end{equation*}
By (i)-(vi), we have $\eta_t(w z^p)\in \WW_t(\g,f,k)$.
\end{proof}

\begin{rem}
Recall that $\{u_\alpha|\alpha \in \bar{S}\}$ and $\{u^\alpha|\alpha \in \bar{S}\}$ are dual bases of $\g.$
We note that  the image of  
\[  - \sum_{\alpha \in \bar{S}}  u_\alpha z^t \, u^\alpha z^t  \in S(\CC[\partial]\otimes \g z^t )\]
by the quotient map $S(\CC[\partial]\otimes \g[z])\to \mathcal{V}_t(\g,f,k)$ 
is an element in $\WW_t(\g,f,k)$.  
\end{rem}

Moreover, we have the following theorem.

\begin{thm}
The fractional W-(super)algebra $\WW_t(\g,f,k)$ associated to a minimal nilpotent $f\in \g$ is isomorphic to the differential algebra of polynomials generated by a basis of the space $G_t:=\bigoplus_{p=0}^{t-1}\g z^p\oplus \g_f z^t.$ Moreover, we have 
\[ \WW_t(\g,f,k)= S\big(\CC[\partial]\otimes \eta_t(G_t)\big),\]
where $\eta_t:G_t \to \WW_t(\g,f,k)$ is the linear map defined by Proposition \ref{Prop: frac_gen}. 
\end{thm}

\begin{proof}
Let $A\in \WW_t(\g,f,k).$ We can find $A'\in S(\CC[\partial]\otimes \eta_t(G_t))$ such that $A-A'$ does not have a term in $S(\CC[\partial]\otimes G_t).$ Hence 
\[ A-A'\quad \in \quad \WW_t(\g,f,k)\cap  \left[\CC[\partial]\otimes \left( x z^t\oplus \bigoplus_{\alpha \in S(1/2)} z_\alpha z^t \right) \right]\cdot \mathcal{V}_t(\g,f,k) .\]
It is not hard to see that 
\[\WW_t(\g,f,k)\cap  \left[\CC[\partial]\otimes \left( x z^t\oplus \bigoplus_{\alpha \in S(1/2)} z_\alpha z^t \right) \right]\cdot \mathcal{V}_t(\g,f,k)=0.\]
Hence $A\in S\big(\CC[\partial]\otimes \eta_t(G_t)\big).$
\end{proof}

The previous proposition and theorem can be restated as follows.

\begin{thm} \label{Thm:frac_gen}
Let $\g$ be a Lie superalgebra with a minimal nilpotent $f$. Then the affine classical fractional W-(super)algebra $\WW_t(\g,f,k)$ is a  differential algebra generated by the following free generators:
\begin{equation}
\eta_t (g z^p)= \left\{ 
\begin{array}{ll}
\eta'_t(g z^p) & \text{ if } g\in \bigoplus_{i=0}^{1} \g ;\\
\eta'_t(g z^p) -k \delta_{p,0} \sum_{\alpha\in S(1/2)
} \partial z^*_\alpha z^t (z_\alpha | w) & \text{ if } g\in \g\left(-\frac{1}{2} \right); \\
\eta'_t (f z^p)-k\big(\partial x z^t -\sum_{\alpha\in S(1/2)} \frac{1}{2}(\partial z^*_\alpha z^t) z_\alpha z^t \big) &  \text{ if } g=f;
\end{array}\right.
\end{equation}
where $p=0, \cdots, t$ and 

\begin{equation}
\eta'_t (g z^p)= \sum_{s\geq 0} \ \sum_{\alpha_1, \cdots, \alpha_s \in  S(1/2)\cup \{0\} }(-1)^{s}\frac{1}{s!}
\left( \prod_{i=1}^s z^*_{\alpha_i} z^t\right) [z_{\alpha_s}, [z_{\alpha_{s-1}}, [\cdots, [z_{\alpha_1}, g z^p]\cdots ]]].
\end{equation}
for $z^*_0=x$ and $z_0=e.$
\end{thm}

Consider the supersymmetric algebra $S(\g_{[t,d]})$ endowed with the Poisson bracket induced by that of $S(\g[z])$ defined as follows:
\begin{equation}
\{a z^p, b z^q\}=  \left\{
\begin{array}{ll}
[a,b] & \text{ if } p=q=0,\\
0 & \text{ if only one of } p,  q \text{ is }0, \\
-[a,b] z^{p+q} & \text{ if } p\neq 0, \ q\neq 0.
\end{array}\right.
\end{equation}

Denote by $\mathcal{I}^{fin}_{[t,d]}$  the ideal of $S(\g_{[t,d]})$ generated by $\{ez^t+1, f z^{t+1}+1\}$ and let 
\begin{equation}
\overline{\eta'_t} (g z^p)= \sum_{s\geq 0} \ \sum_{\alpha_1, \cdots, \alpha_s \in  S(1/2)\cup \{0\} }(-1)^{s}\frac{1}{s!}
\left( \prod_{i=1}^s z^*_{\alpha_i} z^t\right) [z_{\alpha_s}, [z_{\alpha_{s-1}}, [\cdots, [z_{\alpha_1}, g z^p]\cdots ]]]
\end{equation}
be an element in $S(\g_{[t,d]})$.

Then we get the following lemma which is useful to find $\lambda$-brackets between generators of $\WW_t(\g,f,k)$ in Theorem \ref{Thm:frac_gen}.

\begin{lem} \label{Lem:brac_frac_1}
We have the following formula:
\begin{equation}
\{\overline{\eta'}_t (g_1 z^p), \overline{\eta'}_t (g_2 z^q)\} +\mathcal{I}^{fin}_{[t,d]}= \overline{\eta'}_t(\{g_1 z^p, g_2 z^q\})+\mathcal{I}^{fin}_{[t,d]}
\end{equation}
where $g_1 (\text{resp. } g_2) \in \g$ if $p\neq 1$ (resp. $q\neq 1$) and $g_1 (\text{resp. } g_2) \in \bigoplus_{i>-1} \g(i)$ if $p=1$ (resp. $q=1$). 
\end{lem}

\begin{proof}
Using the fact that
 $\{\n z^t,  g z^p\}=\{ x z^t, g z^p\}=0$ if $ g\in \left\{ \begin{array}{ll} \g & \text{ if } p\neq 1, \\  \bigoplus_{i>-1} \g(i) & \text{ if } p=1, \end{array} \right.$ and Leibniz rules, we can show that 
\begin{equation}
\begin{aligned}
 &  \sum_{\alpha_1, \cdots, \alpha_s \in  S(1/2)\cup \{0\} }(-1)^{s}\frac{1}{s!} 
\left( \prod_{i=1}^s z^*_{\alpha_i} z^t\right) [z_{\alpha_s}, [z_{\alpha_{s-1}}, [\cdots, [z_{\alpha_1}, \{g_1 z^p, g_2 z^q\}]\cdots ]]] +\mathcal{I}^{fin}_{[t,d]} \\
& \displaystyle = \sum_{s'=0}^s \frac{1}{s'!(s-s')!}  \mathop{\sum_{\beta_1, \cdots, \beta_{s'} \in  S(1/2)\cup \{0\} }}_{\gamma_1, \cdots, \gamma_{s-s'} \in S(1/2)\cup \{0\}}   \left\{\left( \prod_{i=1}^{s'} z^*_{\beta_i} z^t\right)  [z_{\beta_{s'}}, [z_{\beta_{s'-1}}, [\cdots, [z_{\beta_1}, g_1 z^p]\cdots ]]]\, , \right.\\
 &\left. \left( \prod_{j=1}^{s-s'} z^*_{\gamma_j} z^t\right) [z_{\gamma_{s-s'}}, [z_{\gamma_{s-s'-1}}, [\cdots, [z_{\gamma_1},  g_2 z^q]\cdots ]]]\right\}+\mathcal{I}^{fin}_{[t,d]} .
\end{aligned}
\end{equation}
Hence we can prove our assertion by direct computations.
\end{proof}

\begin{lem} \label{Lem:bracket_fz}
Let $g z^p\in \g_{[t,d]}$ for $g\in \g$ and $p=0, \cdots, t$. Then the following equality holds.
\begin{equation} \label{Eqn:bracket_fz}
\{\overline{\eta'}_t (f z), \overline{\eta'}_t (g z^p)\}+\mathcal{I}^{fin}_{[t,d]} = - \overline{\eta'}_t([e,g z^p])+\overline{\eta'}_t(\{f z, g z^p\})+\mathcal{I}^{fin}_{[t,d]}.
\end{equation}
\end{lem}

\begin{proof}
Observe that 
\begin{equation} \label{Eqn:key_lem_fz}
 \{f z, x z^t\}= -f z^{t+1}\in1+\mathcal{I}^{fin}_{[t,d]}, \quad \{ fz, z^*_\alpha z^t\}=0 \text{ for } \alpha \in S(1/2).
 \end{equation}

The second term in the RHS of (\ref{Eqn:bracket_fz}) follows from the same argument as that of Lemma \ref{Lem:brac_frac_1}.

The first term  in the RHS of (\ref{Eqn:bracket_fz}) follows from (\ref{Eqn:key_lem_fz}). For details, observe that 
\[ x z^t \ z^*_\alpha z^t= z^*_\alpha z^t \ x z^t, \quad [z_\alpha, [e, A]]= [e, [z_\alpha, A]] . \]
Hence, for any $\alpha \in S(1/2)\cup\{0\}$, the term with $x z^t$ in $\overline{\eta'}_t( g z^p)$ can be written as 
\[ \sum_{s\geq 1} \frac{1}{(s-1)!} (-1)^s x z^t  \left(\prod_{i=1}^{s-1}  z^*_{\alpha_i} z^t \right) [ z_{\alpha_{s-1}}, [\cdots, [z_{\alpha_1}, [e, g z^p]]\cdots]].\]
Hence 
\begin{equation}
\begin{aligned}
&  \{\overline{\eta'}_t (f z), \overline{\eta'}_t (g z^p)\}-\overline{\eta'}_t(\{f z, g z^p\}) \\
& =  \sum_{s\geq 1} \{ f z, x z^t \} \frac{1}{(s-1)!} (-1)^s  \left(\prod_{i=1}^{s-1}  z^*_{\alpha_i} z^t \right) [ z_{\alpha_{s-1}}, [\cdots, [z_{\alpha_1}, [e, g z^p]]\cdots]]\\
\end{aligned}
\end{equation}
and 
\begin{equation}
 \{\overline{\eta'}_t (f z), \overline{\eta'}_t (g z^p)\}-\overline{\eta'}_t(\{f z, g z^p\})+\mathcal{I}^{fin}_{[t,d]}  = -\overline{\eta'}_t([e,g z^p])+\mathcal{I}^{fin}_{[t,d]}.
\end{equation}
\end{proof}

\begin{thm}
The $\lambda$-brackets between generators of $\WW_t(\g,f,k)$ are as follows:
\begin{equation*}\begin{array}{ll}
\{\eta_t(g_1)_\lambda \eta_t(g_2)\}= \eta_t([g_1, g_2])+k\lambda(g_1|g_2) & \text{ for } g_1, g_2\in \g; \\
\{\eta_t(fz)_\lambda \eta_t(f)\} = -\eta_t(2x)-k\lambda \\
\{\eta_t(fz)_\lambda \eta_t(g)\}= -\eta_t([e,g]) & \text{ for } g\in \bigoplus_{i>-1}\g(i);\\
\{\eta_t(fz)_\lambda \eta_t(g z^p)\}= -\eta_t([f,g]z^{p+1})-\eta_t([e,g]z^p) & \text{ for } g\in\g,  p\geq 1;\\
\{\eta_t(g_1 z^p)_\lambda \eta_t(g_2 z^q)\} = -\eta_t([g_1, g_2]z^{p+q}) &  \text{ for } g_1 z^p,\, g_2 z^q\in \bigoplus_{i>-1} \g(i)z \oplus \bigoplus_{j>1} \g z^j
\end{array} 
\end{equation*}
\end{thm}

\begin{proof}
The theorem follows from Lemma \ref{Lem:brac_frac_1} and \ref{Lem:bracket_fz}. We have to check  $n$-th products for $n\geq 1$ which can be shown by simple computations. 
\end{proof}

\begin{rem}
In \cite{S}, a Hamiltonian operator on a given classical affine fractional W-algebra associated to a Lie algebra is introduced. Using the Hamiltonian operator, a classical finite fractional W-algebra can be defined. Analogously, we can find a Hamiltonian operator on a classical affine fractional W-superalgebra and a   finite fractional W-superalgebra can be constructed using the operator. 
\end{rem}

\end{document}